\pdfoutput=1 
\documentclass[11pt]{article}%
\DeclareMathAlphabet{\mathbbold}{U}{bbold}{m}{n}
\usepackage{amsmath,amsfonts,amssymb}
\usepackage{amsthm}
\usepackage{graphicx}
\usepackage{fullpage}
\usepackage[pagebackref]{hyperref}
\hypersetup{
    colorlinks=true, 
    linkcolor=blue, 
    citecolor=blue, 
    urlcolor=blue 
}
\usepackage{microtype}
\usepackage{thm-restate}
\usepackage[capitalize,nameinlink]{cleveref} 

\usepackage{tikz}
\usepackage{pgfplots}
\pgfplotsset{compat=1.13}
\usepackage{wrapfig}
\usepackage{enumerate}
\usepackage{wrapfig}
\newcommand{\eat}[1]{}

\renewcommand{\backref}[1]{}

\renewcommand{\backrefalt}[4]{%
\ifcase #1 %
\or
[p.\ #2]%
\else
[pp.\ #2]%
\fi}

\makeatletter
\newcommand{\para}{%
 \@startsection{paragraph}{4}%
 {\z@}{2ex \@plus 3.3ex \@minus .2ex}{-1em}%
 {\normalfont\normalsize\bfseries}%
}
\makeatother

\allowdisplaybreaks

\providecommand{\U}[1]{\protect\rule{.1in}{.1in}}
\newtheorem{theorem}{Theorem}

\newtheorem{corollary}[theorem]{Corollary}

\newtheorem{definition}[theorem]{Definition}

\newtheorem{lemma}[theorem]{Lemma}

\newtheorem{problem}[theorem]{Problem}
\newtheorem{proposition}[theorem]{Proposition}

\let\oldproofname=\proofname
\renewcommand{\proofname}{\rm\bf{\oldproofname}}

\newcommand{\id}{\mathbbold{1}}
\newcommand{\R}{\mathcal{R}}

\newcommand{\<}{\langle}
\renewcommand{\>}{\rangle}
\DeclareMathOperator{\E}{\mathbb{E}}

\newcommand{\AND}{\mathsf{AND}}
\newcommand{\eps}{\varepsilon}
\newcommand{\Reals}{\mathbb{R}}

\renewcommand{\deg}{\mathrm{deg}}

\newcommand{\mathify}[1]{\ifmmode{#1}\else\mbox{$#1$}\fi}

\newcommand{\cl}[1]{\mathsf{#1}}

\newcommand{\ApxCount}{\mathsf{ApxCount}_{N,w}}
\newcommand{\AndApxCount}{\mathsf{AND}_2 \circ \mathsf{ApxCount}_{N,w}}

\begin{document}

\title{\bfseries Quantum Lower Bounds for Approximate\\ 
Counting via Laurent Polynomials%
\footnote{This paper subsumes preprints 
\href{https://arxiv.org/abs/1808.02420}{arXiv:1808.02420} and \href{https://arxiv.org/abs/1902.02398}{arXiv:1902.02398} 
by the first and third authors, respectively.}%
}

\author{
Scott Aaronson\thanks{University of Texas at Austin. \ Email:
\texttt{aaronson@cs.utexas.edu}. \ Supported by a Vannevar Bush Fellowship from the US
Department of Defense, a Simons Investigator Award, and the Simons
\textquotedblleft It from Qubit\textquotedblright\ collaboration.} 
\qquad
Robin Kothari\thanks{Microsoft Quantum and Microsoft Research. Email: \texttt{robin.kothari@microsoft.com}.}
\qquad
William Kretschmer\thanks{University of Texas at Austin. \ Email:
\texttt{kretsch@cs.utexas.edu}. \ Supported by a Vannevar Bush Fellowship from the US
Department of Defense and a Simons Investigator Award..}
\qquad
Justin Thaler\thanks{Georgetown University. Email: \texttt{justin.thaler@georgetown.edu}. Supported by NSF CAREER award CCF-1845125.}
}

\date{}
\maketitle

\begin{abstract}
 
We study quantum algorithms that are given
access to trusted and untrusted quantum witnesses.
We establish strong
limitations of such algorithms,
via new techniques based on \emph{Laurent polynomials} (i.e., polynomials with positive and negative integer exponents). 
Specifically, we resolve the complexity of \emph{approximate counting}, the problem of multiplicatively estimating the size of a nonempty set $S \subseteq [N]$, in two natural generalizations of quantum query complexity.

Our first result holds in the standard Quantum Merlin--Arthur 
($\mathsf{QMA}$) setting, in which a quantum algorithm receives
an untrusted quantum witness. 
We show that, if the algorithm makes $T$\ quantum queries
to $S$, and also receives an (untrusted) $m$-qubit quantum witness, 
then either $m = \Omega(|S|)$ or 
$T=\Omega \bigl(\sqrt{N/\left\vert S\right\vert } \bigr)$.
This is optimal, matching the straightforward protocols where 
the witness is either empty, or specifies all the elements of $S$.
As a corollary, this resolves the open problem of giving an oracle
separation between $\mathsf{SBP}$, the complexity class that captures
approximate counting, and $\mathsf{QMA}$.

In our second result, we ask what if, in addition to a membership oracle for $S$, 
 a quantum algorithm is also given \textquotedblleft 
QSamples\textquotedblright ---i.e., copies of the state $\left\vert
S\right\rangle = \frac{1}{\sqrt{\left\vert S\right\vert }} \sum_{i\in S}|i\>$---
or even access to a unitary transformation that enables QSampling? \ We show that, even then, the
algorithm needs either $\Theta \bigl(\sqrt{N/\left\vert S\right\vert }\bigr)$%
\ queries or else $\Theta \bigl(\min \bigl\{\left\vert S\right\vert ^{1/3},%
\sqrt{N/\left\vert S\right\vert }\bigr\}\bigr)$\ QSamples or accesses to the unitary. \

Our lower bounds in both settings make essential use of Laurent polynomials,
but in different ways.

\end{abstract}

\clearpage
\tableofcontents
\clearpage

\section{Introduction}
\label{INTRO}\label{sec:intro}

Understanding the power of quantum algorithms has been a central research goal over the last few decades.
One success story in this regard has been the discovery of powerful methods that establish limitations on quantum algorithms in the standard setting of \emph{query complexity}. This setting roughly asks, for a
specified function $f$, how many
bits of the input must be examined by any quantum algorithm that computes $f$ (see \cite{bw} for a survey of query complexity).

A fundamental topic of study in complexity theory is algorithms that are ``augmented'' with additional information, such as an untrusted witness provided by a powerful prover.
For example, the classical complexity class $\mathsf{NP}$ is defined this way.
In the quantum setting, if we go beyond standard query algorithms, and allow algorithms to receive a quantum state, the model becomes much richer, and we have very few techniques to establish lower bounds for these algorithms.
In this paper, we develop such techniques.
Our methods crucially use \emph{Laurent polynomials}, which are polynomials with positive and negative integer exponents.

We demonstrate the power of these lower bound techniques by proving optimal lower bounds for the \emph{approximate counting} problem, which captures the following task.
Given a nonempty finite set $S\subseteq [N] :=\left\{
1,\ldots ,N\right\} $, estimate its cardinality, $\left\vert S\right\vert $, 
to within some constant (say, 2) multiplicative accuracy.
Approximate counting is a fundamental task with a rich
history in computer science. This includes the works of Stockmeyer~\cite%
{Sto85}, which showed that approximate counting is in the polynomial
hierarchy, and Sinclair and Jerrum~\cite{sinclairjerrum}, which showed the
equivalence between approximate counting and approximate sampling that
enabled the development of a whole new class of algorithms based on Markov chains.
Additionally, approximate counting precisely highlights the limitations of current lower bound techniques for the complexity class $\mathsf{QMA}$ 
(as we explain in \Cref{sec:QMAintro}).

Formally, we study the following decision version of the problem in this paper:

\begin{problem}[Approximate Counting]
In the $\mathsf{ApxCount}_{N,w}$ problem, our goal is to decide whether a
nonempty set $S\subseteq \lbrack N]$ satisfies $\left\vert S\right\vert \geq
2w$ (YES) or $\left\vert S\right\vert \leq w$ (NO), promised that one of
these is the case.
\end{problem}

In the query model, the algorithm is given a membership oracle for $S$: one that, for any $i\in \left[ N\right] $, returns whether $%
i\in S$. \ How many queries must we make, as a function of both $N$\ and $%
\left\vert S\right\vert $, to solve approximate counting with high
probability?

For classical randomized algorithms, it is easy to see that $\Theta (N/|S|)$\
membership queries are necessary and sufficient. 
For quantum algorithms, which can query the membership oracle on superpositions
of inputs, Brassard et al.~\cite{bht:count,BHMT02} gave an algorithm
that makes only $O\bigl(\sqrt{N/|S|}\bigr)$\ queries.
It follows from the optimality of Grover's algorithm (i.e., the BBBV Theorem \cite{bbbv}) that this cannot be improved. Hence, the classical and quantum complexity of approximate counting with membership queries alone is completely understood.
In this paper, we study the complexity of approximate counting in models with untrusted and trusted quantum states.

\subsection{First result: QMA complexity of approximate counting}
\label{sec:QMAintro}

Our first result, presented in \Cref{sec:SBPQMA}, considers the standard Quantum Merlin--Arthur ($\mathsf{QMA}$) setting, in which the quantum algorithm receives an untrusted quantum state (called the witness). 
This model is the quantum analogue of the classical complexity class $\mathsf{NP}$, and is of great
interest in quantum complexity theory. It captures natural problems about ground states of physical systems, properties of quantum circuits and channels, noncommutative constraint satisfaction problems, consistency of representations of quantum systems, and more~\cite{Boo13}.

In a $\mathsf{QMA}$ protocol, a skeptical verifier
(Arthur) receives a quantum witness state $|\psi \>$ from an all-powerful but
untrustworthy prover (Merlin), in support of 
the claim that $f(x)=1$. 
\ Arthur then needs to verify $|\psi \>$, via some
algorithm that satisfies the twin properties of \textit{completeness} and
\textit{soundness}. \ That is, if $f(x)=1$, then there must exist some $|\psi \>$\ that causes Arthur
to accept with high probability, while if $f(x)=0$, then every $|\psi \>$\ must cause Arthur to reject with high probability. \ We call such a protocol a $\mathsf{QMA}$ (Quantum Merlin--Arthur) protocol for computing $f$. 

In the query complexity setting, there are two resources to consider: the length of the quantum witness, $m$, and
the number of queries, $T$, that Arthur makes to the membership oracle. \ 
A $\mathsf{QMA}$ protocol for $f$ is efficient if both $m$ and $T$
are $\mathrm{polylog}(N)$. 

\para{The known lower bound technique for $\mathsf{QMA}$.} 
Prior to our work, all known $\mathsf{QMA}$ lower bounds used the same proof technique.\footnote{%
There is one special case in which it is trivial to lower-bound $\mathsf{QMA}$ complexity. 
Consider the $\AND_N$ function on $N$ bits that outputs 1 if and only if all $N$ bits equal $1$. For this function, since Merlin wants to convince Arthur that $f(x)=1$, intuitively there is nothing interesting that Merlin can say to Arthur other than ``$x$ is all ones'' since that is the only input with $f(x)=1$. Formally, Arthur can simply create the witness state that an honest Merlin would have sent on the all ones input, and hence Arthur does not need Merlin \cite{razshpilka}.
For such functions, $\mathsf{QMA}$ complexity is the same as standard quantum query complexity.
} 
The technique establishes (and exploits) the complexity class containment 
$\mathsf{QMA}\subseteq \mathsf{SBQP}$, where $\mathsf{SBQP}$\ is
a complexity class that models quantum algorithms with tiny acceptance and
rejection probabilities. 
Specifically, we say that a function $f$ has $\mathsf{SBQP}$ query complexity at most $k$
if there exists a $k$-query quantum algorithm that
\begin{itemize}
\item outputs $1$ with probability $\geq \alpha $\ when $f(x)=1$, and
\item outputs $1$ with probability $\leq \alpha /2$\ when $f(x)=0$,
\end{itemize}
\noindent for some $\alpha $ that does not depend on the input (but may depend on the input size).
Note
that when $\alpha =2/3$, we recover standard quantum query complexity. \ But
$\alpha $ could be also be exponentially small, which makes $\mathsf{SBQP}$
algorithms very powerful.

Nevertheless, one can establish significant limitations on $\mathsf{SBQP}$
algorithms, by using a variation of the polynomial method of Beals et al.\
\cite{bbcmw}. \ If a function $f$ can be evaluated by an $\mathsf{SBQP}$
algorithm with $k$ queries, then there exists a real polynomial $p$ of
degree $2k$ such that $p(x)\in \lbrack 0,1]$ whenever $f(x)=0$ and $p(x)\geq
2$ whenever $f(x)=1$. \ The minimum degree of such a polynomial is also
called \emph{one-sided approximate degree}~\cite{BT15}.

The relationship between $\mathsf{SBQP}$ and $\mathsf{QMA}$ protocols is
very simple: if $f$ has a $\mathsf{QMA}$ protocol that receives an $m$-qubit
witness\ and makes $T$ queries, then it also has an $\mathsf{SBQP}$
algorithm that makes $O(mT)$ queries. \ This was essentially observed by Marriott and 
Watrous~\cite[Remark 3.9]{marriott} and used by Aaronson \cite%
{aaronson_szk} to show an oracle relative to which $\mathsf{SZK} \not\subset \mathsf{QMA}$. 

\para{Beyond the known lower bound technique for $\mathsf{QMA}$.} 
Our goal is to find a new method of lower bounding $\mathsf{QMA}$, that does not go through $\mathsf{SBQP}$  complexity. 
The natural way to formalize this quest is to find a problem that has an efficient 
$\mathsf{SBQP}$ algorithm, and show that it
does not have an efficient $\mathsf{QMA}$ protocol.
A natural candidate for this is the $\mathsf{ApxCount}_{N,w}$ problem.
We know that $\mathsf{ApxCount}_{N,w}$ \textit{does} have a very simple $\mathsf{SBQP}$ algorithm of cost 1:
 the algorithm picks an $i\in \left[ N\right] $\ uniformly at random, and accepts if and only if $i\in S$. \ 
Clearly the algorithm accepts with probability greater than $2w/N$ on yes inputs and with probability at most $w/N$ on no inputs.

Our first result establishes that $\mathsf{ApxCount}_{N,w}$ 
does \emph{not} have an efficient $\mathsf{QMA}$ protocol.

\begin{restatable}{theorem}{qmabound}
\label{thm:qmabound}
Consider a $\mathsf{QMA}$ protocol that solves $\ApxCount$. If the protocol receives a quantum witness of length $m$, and makes $T$ queries to the membership oracle for $S$, then
either $m = \Omega(w)$ or 
$ T = \Omega\bigl(\sqrt{N/w}\bigr)$.
\end{restatable}

This lower bound proved in \Cref{sec:QMA} resolves the $\mathsf{QMA}$ complexity of  $\ApxCount$, as (up to a $\log N$ factor) it matches the cost of two trivial $\mathsf{QMA}$ protocols.
In the first, Merlin sends $2w$ items claimed to be in $S$,
and Arthur picks a constant number of the items at random
and confirms they are all in $S$ with one membership query each. 
This protocol has witness length $m=O(w\log N)$ (the number of bits needed to specify $2w$ elements out of $N$) and $T=O(1)$. 
In the second protocol, Merlin does nothing, and Arthur solves the problem with $T=O\bigl(\sqrt{N/w}\bigr)$ quantum queries.


\setlength{\columnsep}{1.5em}
\setlength{\intextsep}{3ex}
\begin{wrapfigure}{r}{0.245\textwidth}
    \centering
    \begin{tikzpicture}[x=1cm,y=1cm]

    \node (MA) at(2,0){$\cl{MA}$};
    \node (QMA) at(1,1){$\cl{QMA}$};
    \node (SBP) at(3,1){$\cl{SBP}$};
    \node (SBQP) at(2,2){$\cl{SBQP}$};
    \node (PP) at(2,3){$\cl{PP}$};    
    \node (AM) at(4,2){$\cl{AM}$};
    
    \path[-] (MA) edge (QMA);
    \path[-] (QMA) edge (SBQP);
    \path[-] (MA) edge (SBP);     
    \path[-] (SBP) edge (SBQP);
    \path[-] (SBP) edge (AM);     
    \path[-] (SBQP) edge (PP);
    
    \end{tikzpicture}
    
    \caption{Relationships between complexity classes. 
    An upward line indicates that a complexity class 
    is contained in the one above it relative to 
    all oracles.\label{fig:relations}\vspace{2ex}}
\end{wrapfigure}
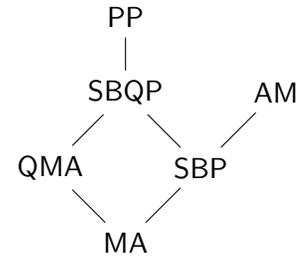

\para{Oracle separation.}
Our result also yields new oracle separations.
The approximate counting problem is complete for the complexity class $\mathsf{SBP}$ \cite{BGM06}, which is sandwiched between $\mathsf{MA}$ (Merlin--Arthur) and $\mathsf{AM}$\ (Arthur--Merlin). \ The class $\mathsf{SBQP}$ (discussed above), first defined by Kuperberg \cite{kuperberg}, is a
quantum analogue of $\mathsf{SBP}$ that contains both $\mathsf{SBP}$ and $\mathsf{QMA}$.

By the usual connection between oracle separations and query complexity
lower bounds, \Cref{thm:qmabound} implies an oracle separation
between $\mathsf{SBP}$ and $\mathsf{QMA}$---i.e., there exists an oracle $A$
such that $\mathsf{SBP}^{A}\not\subset \mathsf{QMA}^{A}$  
(see \Cref{cor:qma_separation}). Prior to our
work, it was known that there exist oracles $A,B$ such that $\mathsf{SBP}%
^{A}\not\subset \mathsf{MA}^{A}$~\cite{BGM06} and $\mathsf{AM}%
^{B}\not\subset \mathsf{QMA}^{B}$, which follows from $\mathsf{AM}^B \not\subset \mathsf{PP}^B$ \cite{vereshchagin}, 
but the relation
between $\mathsf{SBP}$\ and\ $\mathsf{QMA}$\ remained elusive.\footnote{%
It is interesting to note that in the non-relativized world, under plausible
derandomization assumptions~\cite{amnp}, we have $\mathsf{NP}=\mathsf{MA}=%
\mathsf{SBP}=\mathsf{AM}$. In this scenario, all these classes are equal,
and all are contained in $\mathsf{QMA}$.} \Cref{fig:relations} 
shows the known inclusion relations among these
classes (all of which hold relative to all oracles).

Previous techniques were inherently
unable to establish this oracle separation for the reason stated above: all existing $\mathsf{QMA}$ lower bounds
intrinsically apply to $\mathsf{SBQP}$ as well. Since $\mathsf{SBP}$ is contained in $\mathsf{SBQP}$, prior techniques cannot establish $\mathsf{SBP}^{A}\not\subset \mathsf{QMA}^{A}$, or even $\mathsf{SBQP}^{A}\not\subset \mathsf{QMA}^{A}$, for any oracle $A$.
Our analysis also yields the first oracle with respect to
which $\mathsf{SBQP}$ is not closed under intersection.



\para{Proof overview.} \ To get around the
issue of $\mathsf{ApxCount}_{N,w}$\ being in $\mathsf{SBQP}$,\ we use a
clever strategy that was previously used by G\"{o}\"{o}s et al.~\cite%
{GLMWZ16}, and that was suggested to us by Thomas Watson (personal
communication).
Our strategy exploits a structural property of $\mathsf{QMA}$: 
the fact that $\mathsf{QMA}$\ is closed under intersection, 
but (at least relative to oracles, and as we'll show) $\mathsf{SBQP}$\ is not.

Given a function $f$, let $\mathsf{AND}_{2}\circ f$ be the
$\mathsf{AND}$ of two copies of $f$ on separate inputs.\footnote{Because we focus on lower bounds, for a promise problem $f$ (such as $\ApxCount$), we take the promise for $\mathsf{AND}_2 \circ f$ to be that both instances of $f$ must satisfy $f$'s promise. Then, any lower bound also applies to more relaxed definitions, such as only requiring one of the two instances to be in the promise.} 
\ Then if $f$ has small $\mathsf{QMA}$ query complexity, it's not hard to see that $\mathsf{AND}%
_{2}\circ f$ does as well:\ Merlin simply sends witnesses corresponding to
both inputs; then Arthur checks both of them independently. \ While it's not
completely obvious, one can verify that a dishonest Merlin would gain
nothing by entangling the two witness states. \ Hence if $\mathsf{ApxCount}%
_{N,w}$ had an efficient $\mathsf{QMA}$ protocol, then so would $\mathsf{AND}%
_{2}\circ \mathsf{ApxCount}_{N,w}$, with the witness size and query
complexity increasing by only a constant factor.

By contrast, even though $\mathsf{ApxCount}_{N,w}$ does have an
efficient $\mathsf{SBQP}$ algorithm, we will show that $\mathsf{AND}%
_{2}\circ \mathsf{ApxCount}_{N,w}$ does not. \ This is the technical core of
our proof and proved in \Cref{sec:SBQP}.


\begin{restatable}{theorem}{sbqp}
\label{thm:sbqp}
Consider an $\mathsf{SBQP}$ algorithm for $\AndApxCount$ that makes $T$ queries to membership oracles for the two instances of $\ApxCount$. Then $T=\Omega\left(\min\bigl\{w,\sqrt{N/w}\bigr\}\right)$.
\end{restatable}

\Cref{thm:sbqp} is quantitatively optimal, as we'll exhibit
a matching $\mathsf{SBQP}$ upper bound. \ 
Combined with the
connection between $\mathsf{QMA}$ and $\mathsf{SBQP}$, \Cref{thm:sbqp}
immediately implies a $\mathsf{QMA}$ lower bound for $\AndApxCount$, and by extension $\ApxCount$ itself. However, this $\mathsf{QMA}$ lower bound is not quantitatively optimal. To obtain the optimal bound of \Cref{thm:qmabound},
we exploit additional analytic properties of 
the $\mathsf{SBQP}$ protocols that are
derived from $\mathsf{QMA}$ protocols.

At a high level, the proof of \Cref{thm:sbqp} assumes that there's an
efficient $\mathsf{SBQP}$ algorithm for $\mathsf{AND}_{2}\circ \mathsf{%
ApxCount}_{N,w}$. \ This assumption yields a low-degree one-sided
approximating polynomial for the problem in $2N$ Boolean variables, where $N$
variables come from each $\mathsf{ApxCount}_{N,w}$\ instance. \ We then
symmetrize the polynomial (using the standard Minsky--Papert symmetrization argument~\cite{mp})
to obtain a bivariate polynomial in two variables $%
x$ and $y$ that represent the Hamming weight of the original instances.\footnote{%
The term ``symmetrization'' originally referred to the process of averaging a multivariate polynomial over permutations of its inputs to obtain a symmetric polynomial. More recently, authors have used ``symmetrization'' more generally to refer to any method for turning a multivariate polynomial into a univariate one in a degree non-increasing manner (see, e.g., \cite{sherstov2009halfspaces, sherstov2010halfspaces}). In this paper, we use the term ``symmetrization'' in this more general sense.
} \
This yields a polynomial $p(x,y)$ that for \emph{integer pairs} $x, y$ (also called lattice points) satisfies
$p(x,y)\in [0,1]$ when either $x \in \{0,\ldots,w\}$ and $y\in \{0,\ldots,w\} \cup \{2w,\ldots,N\}$, or
(symmetrically) $y \in \{0,\ldots,w\}$ and $x\in \{0,\ldots,w\} \cup \{2w,\ldots,N\}$. \ If
both $x \in \{2w,\ldots,N\}$ and $y \in \{2w,\ldots,N\}$, then $p(x,y)\geq 2$. 
This polynomial $p$ is depicted in \Cref{fig:introL}.

\begin{figure}[tbp]
\centering
\begin{tikzpicture}[y=0.35cm, x=0.35cm]
	\pgfmathsetmacro{\N}{20};
	\pgfmathsetmacro{\w}{4};
	\pgfmathsetmacro{\tw}{\w+\w};
	\pgfmathsetmacro{\Np}{\N+1};

	\fill[gray,opacity=.4] (\tw,0) -- (\tw,\w) -- (\N,\w) -- (\N,0);
	\fill[gray,opacity=.4] (0,\tw) -- (\w,\tw) -- (\w,\N) -- (0,\N);
	\fill[gray,opacity=.4] (0,0) -- (\w,0) -- (\w,\w) -- (0,\w);
	\fill[gray,opacity=.1] (\tw,\tw) -- (\tw,\N) -- (\N,\N) -- (\N,\tw);

	\draw[->] (0,0) -- coordinate (x axis mid) (\Np,0);
    \draw[->] (0,0) -- coordinate (y axis mid) (0,\Np);
    	
    \draw (0,1pt) -- (0,-3pt) node[anchor=north] {0};
    \draw (\w,1pt) -- (\w,-3pt) node[anchor=north] {$w$\vphantom{l}};
    \draw (\tw,1pt) -- (\tw,-3pt) node[anchor=north] {$2w$};
    \draw (\N,1pt) -- (\N,-3pt) node[anchor=north] {$N$};

    \draw (1pt,0) -- (-3pt,0) node[anchor=east] {0};
    \draw (1pt,\w) -- (-3pt,\w) node[anchor=east] {$w$};
    \draw (1pt,\tw) -- (-3pt,\tw) node[anchor=east] {$2w$};
    \draw (1pt,\N) -- (-3pt,\N) node[anchor=east] {$N$};

	\node[below=0.8cm] at (x axis mid) {$x$};
	\node[rotate=90, above=0.8cm] at (y axis mid) {$y$};
	\node at (\w+\N/2, \w/2) {$p(x, y) \in [0,1]$};
	\node[rotate=90] at (\w/2, \w+\N/2) {$p(x, y) \in [0,1]$};
	\node at (\w/2, \w/2) [align=left]{$p(x, y)$ \\ $\in [0,1]$};
	\node at (\w+\N/2, \w+\N/2) {$p(x, y) \geq 2$};

     \draw[thick,domain=1:{\N/(2*\w)},smooth,variable=\t,blue] plot ({2*\w*\t},{2*\w/\t});
     \draw[thick,domain=1:{\N/(2*\w)},smooth,variable=\t,blue] plot ({2*\w/\t},{2*\w*\t});

\end{tikzpicture}
\caption{The behavior of the (Minsky--Papert symmetrized) bivariate polynomial $p(x,y)$ at integer points $(x,y)$ in the proof of \Cref{thm:sbqp}. The polynomial $q$ obtained by erase-all-subscripts symmetrization is not depicted. We later restrict $q$ to a hyperbola similar to the one drawn in blue.}
\label{fig:introL}
\end{figure}
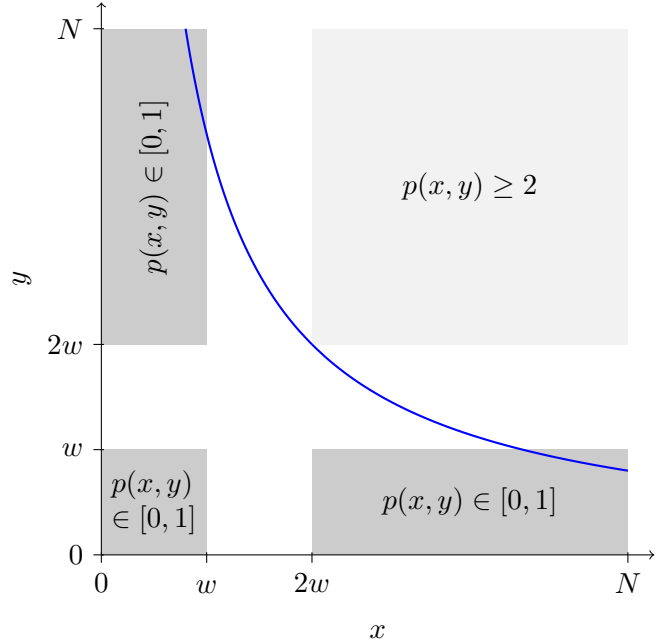

One difficulty is that we have a guarantee on the behavior of $p$ at lattice 
points only, whereas the rest of our proof requires precise control over 
the polynomial even at non-integer points. We ignore this issue for now and assume that $p(x,y)\geq 2$ for all
real values $x, y \in [2w, N]$, and $p(x, y) \in [0, 1]$
whenever $x \in [0, w]$ and $y \in [2w, N]$ 
or vice versa. We outline how we 
address integrality issues one
paragraph hence.

The key remaining difficulty is that we want to lower-bound the degree of 
a bivariate polynomial, but almost all known
lower bound techniques apply only to univariate polynomials.
To address this, we introduce a new technique
to reduce the number of variables (from $2$ to $1$)
in a degree-preserving way: we pass a \textit{hyperbola}
through the $xy$ plane (see \Cref{fig:introL}) and consider 
the polynomial $p$ restricted to the hyperbola.\ 
Doing so gives us a new univariate \emph{Laurent} polynomial $\ell (t) = p(2wt, 2w/t)$, 
whose positive and negative degree is at most $\deg(p)$. \
This Laurent polynomial has an additional symmetry, 
which stems from the fact that $\mathsf{AND}%
_{2}\circ \mathsf{ApxCount}_{N,w}$ is the $\mathsf{AND}$ of two identical
problems (namely, $\mathsf{ApxCount}_{N,w}$). \ We leverage this symmetry to
view $\ell (t)$, a Laurent polynomial in $t$, as an ordinary univariate
polynomial $r$ in $t+1/t$ of degree $\deg(p)$. 
We know that $r(2)=\ell(1)=p(2w, 2w)\geq 2$,
while for all $k \in [2.5, N/w+w/N]$, we know that
$r(k) \in [0, 1]$.
It then follows from classical
results in approximation theory that this univariate polynomial
must have degree $\Omega\bigl(\sqrt{N/w}\bigr)$.

\
Returning to integrality issues, to obtain a polynomial
whose behavior we can control at non-integer points, 
we use a different symmetrization argument (dating
back at least to work of Shi \cite{shi})
 that we call 
\textquotedblleft erase-all-subscripts\textquotedblright\ symmetrization (see \cref{lem:eraseallsubscripts}). 
This symmetrization yields a bivariate polynomial $q$ of the same degree as $p$ that is bounded in $[0,1]$ at all \emph{real-valued}
inputs in $[0, N] \times [0, N]$. However, while
we have more control over $q$'s values at non-integer
inputs relative to $p$,
we have
\emph{less} control over $q$'s values at integer
inputs relative to $p$, and
this introduces additional challenges. (These
challenges are not merely annoyances; they are 
why the $\mathsf{SBQP}$ complexity of $\AndApxCount$
is $T=\Theta\bigl(\min\bigl\{w,\sqrt{N/w}\bigr\}\bigr)$,
and not $\Theta\bigl(\sqrt{N/w}\bigr)$).
Ultimately,
both types of symmetrization play an important role in our analysis, as we use $p$ to
bound $q$ when the polynomials have degree $o(w)$, using tools from approximation
theory and Chernoff bounds.

\subsection{Second result: Approximate counting with quantum samples}
\label{s:secondresultintro}
Our second result resolves the complexity of $\ApxCount$ in a different generalization of the quantum query model, in which the algorithm is given access to certain (trusted) quantum states.

\para{Quantum samples.} In practice, when
trying to estimate the size of a set $S\subseteq\left[ N\right] $, often we
can do more than make membership queries to $S$. \ At the least, often we
can efficiently generate nearly uniform \textit{samples} from $S$, for
instance by using Markov Chain Monte Carlo techniques. \ To give two
examples, if $S$ is the set of perfect matchings in a bipartite graph, or
the set of grid points in a high-dimensional convex body, then we can
efficiently sample $S$ using the seminal algorithms of Jerrum, Sinclair, and
Vigoda \cite{jsv} or of Dyer, Frieze, and Kannan \cite{dfk},\ respectively.

The natural quantum generalization of uniform sampling from a set $S$ is \emph{QSampling} $S$---a term coined in 2003 by Aharonov and Ta-Shma \cite{at}, and which means that we can approximately prepare the uniform superposition%
\begin{equation}
\left\vert S\right\rangle := \frac{1}{\sqrt{\left\vert S\right\vert }}%
\sum_{i\in S}\left\vert i\right\rangle
\end{equation}
via a polynomial-time quantum algorithm (where \textquotedblleft
polynomial\textquotedblright\ here means $\mathrm{polylog}(N)$). \ Because
we need to uncompute garbage, the ability to prepare $\left\vert
S\right\rangle $ as a coherent superposition is a more stringent requirement
than the ability to classically sample from $S$. 
Indeed,  Aharonov and Ta-Shma \cite{at}
showed that the ability to QSample lends considerable power:
all problems in the complexity class $\mathsf{SZK}$ (which 
contains problems that are widely believed
 be hard on average \cite{goldreich1993perfect, goldwasser1989knowledge, micciancio2003statistical, goldreich2000limits, peikert2008noninteractive})
can be efficiently reduced to the task of \emph{QSampling} some set
that can be \emph{classically} sampled in polynomial time.
To be clear, QSampling supposes that the algorithm
is given trusted copies of $|S\>$; unlike in the $\mathsf{QMA}$ setting, the state need not be ``verified'' by the algorithm.

On the other hand, Aharonov and Ta-Shma \cite{at},\ and Grover and Rudolph
\cite{groverrudolph}, observed that many interesting sets $S$\ can be
efficiently QSampled as well.\footnote{In particular, this holds for all sets $S$\ such that we
can approximately count not only $S$ itself, but also the restrictions of $S$
obtained by fixing bits of its elements. 
So in
particular, the set of perfect matchings in a bipartite graph, and the set
of grid points in a convex body, can both be efficiently QSampled.
There
are other sets that can be QSampled but not because of this reduction. \ A
simple example would be a set $S$\ such that $\left\vert S\right\vert \geq%
\frac{N}{\mathrm{polylog}N}$: in that case we can efficiently prepare $%
\left\vert S\right\rangle $\ using postselection, but approximately counting
$S$'s restrictions might be hard.}


\para{QSampling via unitaries.} 
In many applications (such as when $S$ is the set of perfect matchings in a bipartite graph or grid points in a convex body),
the reason an algorithm can QSample $S$ is because
it is possible to efficiently construct a quantum circuit
implementing a unitary operator $U$ that
prepares the state $|S\>$. 
Access to this unitary $U$ potentially conveys
substantially more power than QSampling alone.
For example, access to $U$ conveys (in a black box manner)
the ability not only to QSample, but also to perform
reflections about $|S\>$: that is, to apply the unitary
transformation
\begin{equation}
\mathcal{R}_{S}:=\mathbbold{1}-2|S\>\langle S|,
\end{equation}%
which has eigenvalue $-1$ for $|S\>$ and eigenvalue $+1$ for all states
orthogonal to $|S\>$. \ 
More concretely, let $U$ be the unitary that performs the
map $U|0\>=|S\>$, for some canonical starting state $|0\>$. \ Since we know
the circuit $U$, we can also implement $U^{\dagger }$, by reversing the
order of all the gates and replacing all the gates with their adjoints. Then
$\mathcal{R}_{S}$ is simply
\begin{equation}
\mathcal{R}_{S}=\mathbbold{1}-2|S\>\langle S|=U\left( \mathbbold{1}%
-2|0\>\langle 0|\right) U^{\dagger }.
\end{equation}

Note that \textit{a priori}, QSamples and reflections about $|S\>$ could be
incomparable resources; it is not obvious how to simulate either one using
the other. \ On the other hand, it is known how to apply a quantum channel
that is $\varepsilon $-close to $\mathcal{R}_{S}$ (in the diamond norm)
using $\Theta (1/\varepsilon )$ copies of $|S\>$~\cite{lmr:pca,KLL+17}.

Access to a quantum circuit computing $U$ also permits an algorithm to efficiently apply $U$ on inputs
that do not produce the state $|S\>$, to construct
a controlled version of $U$, etc.

\para{Results.} As previously mentioned, 
Aharonov and Ta-Shma \cite{at}
showed that the ability to QSample lends considerable power,
including the ability to efficiently solve $\mathsf{SZK}$-complete problems. 
It is natural to ask just how much power the ability to QSample conveys. 
In particular, can one extend the result of Aharonov and Ta-Shma \cite{at}
from any problem in $\mathsf{SZK}$ to any problem in $\mathsf{SBP}$? 
Equivalently stated, can one
solve approximate counting efficiently,
using \textit{any} combination of $\mathrm{polylog}(N)$ queries and applications of a unitary $U$ that permits QSampling?\footnote{%
We thank Paul Burchard (personal communication) for bringing this question
to our attention.}
In this work, we show that the answer is no.
We begin by focusing on the slightly simplified
setting where the algorithm
is only permitted to perform membership queries,
QSamples, and reflections about the state $|S\>$.

\begin{restatable}{theorem}{main}
\label{thm:main}
Let $Q$ be a quantum algorithm that makes $T$ queries to the membership oracle for $S$,
and uses a total of $R$ copies of $|S\>$ and reflections about $|S\>$.
If $Q$ decides whether $\left\vert S\right\vert =w$\ or $\left\vert S\right\vert =2w$
with high probability, promised that one of those is the case, then either
\begin{equation}
T=\Omega\left(  \sqrt{\frac{N}{w}}\right)  \qquad \textrm{or} \qquad
R=\Omega\left(  \min\left\{  w^{1/3},\sqrt{\frac{N}{w}}\right\}  \right).
\end{equation}
\end{restatable}

This is proved in \Cref{sec:dualpoly}. 
So if (for example) we set $w:=N^{3/5}$, then any quantum algorithm must
either query $S$, or use the state $\left\vert S\right\rangle $\ or
reflections about $\left\vert S\right\rangle $, at least $\Omega (N^{1/5})$\
times. 
Put another way, \Cref{thm:main} means that unless $w$ is very small ($w \leq \mathrm{polylog}(N))$) or extremely large ($w\geq N/\mathrm{polylog}(N)$), the ability to QSample $S$, reflect about $|S\>$, and determine membership in $S$ is not sufficient to approximately count $S$ efficiently. Efficient quantum algorithms for approximate counting will have to leverage additional structure of $S$, beyond the ability to QSample, reflect about $|S\>$, and determine membership in $S$.

In \Cref{thm:unitary} of \Cref{sec:unitary}, we then strengthen \Cref{thm:main} to hold
not only against algorithms
that can QSample
and reflect about $|S\>$ (in addition
to performing membership queries to $S$),
but also against all algorithms that are given 
access to
a specific unitary $U$ that
conveys the power to QSample
and reflect about $|S\>$.\footnote{To be precise,
the unitary $U$ to which the lower bound of \Cref{thm:unitary} applies maps a canonical starting state to $|S\>|S\>$.
As we explain in \Cref{sec:unitary}, such a unitary suffices to implement QSampling, reflections
about $|S\>$, etc., 
since the register containing the second
copy of $|S\>$ can simply be ignored.}

\medskip
Finally, we prove that the lower bounds in \Cref{thm:main} and \Cref{thm:unitary} are optimal. \ As
mentioned before, Brassard et al.~\cite{bht:count} gave a quantum algorithm
to solve the problem using $T=O(\sqrt{{N}/{w}})$\ queries alone, which
proves the optimality of the lower bound on the number of queries. On the other hand, it's easy to solve the problem using $O\left( \sqrt{w}%
\right) $\ copies of $\left\vert S\right\rangle $ alone, by simply measuring
each copy of $\left\vert S\right\rangle $\ in the computational basis and
then searching for birthday collisions. \ Alternately, we can solve the
problem using $O\bigl(\frac{N}{w}\bigr)$\ copies of $\left\vert
S\right\rangle $ alone, by projecting onto the state $|\psi \>=\frac{1}{%
\sqrt{N}}\left( \left\vert 1\right\rangle +\cdots +\left\vert N\right\rangle
\right) $ or its orthogonal complement. \ This measurement succeeds with
probability $|\langle S|\psi \>|^{2}=\frac{|S|}{N}$, so we can approximate $|S|$ by
simply counting how many measurements succeed.

In \Cref{UPPER} we improve on these algorithms by using samples \emph{and}
reflections, and thereby establish that \Cref{thm:main} and \Cref{thm:unitary} are tight.

\begin{restatable}{theorem}{alg}
\label{thm:alg}
There is a quantum algorithm that solves $\ApxCount$ with high probability using
$R$ copies of $|S\>$ and reflections about $|S\>$, where $R = O\left( \min \left\{ w^{1/3} , \sqrt{\frac{N}{w}} \right\} \right)$.
\end{restatable}

\para{The Laurent polynomial method.}In our
view, at least as interesting as \Cref{thm:main} is the technique used to
achieve it. \ In 1998, Beals et al.\ \cite{bbcmw}\ famously observed that, if
a quantum algorithm $Q$\ makes $T$ queries to an input $X$, then $Q$'s
acceptance probability can be written as a real multilinear polynomial in
the bits of $X$, of degree at most $2T$. \ And thus, crucially, if we want
to \textit{rule out} a fast quantum algorithm to compute some function $%
f(X) $, then it suffices to show that any real polynomial $p$\
that approximates $f$\ pointwise must have high degree. \ This general
transformation, from questions about quantum algorithms to questions about
polynomials, has been used to prove many results that were not known
otherwise at the time, including the quantum lower bound for the collision
problem \cite{aar:col,as}\ and the first direct product theorems for quantum
search \cite{aar:adv,ksw}.

In our case, even in the simpler model with only queries and samples (and no
reflections), the difficulty is that the quantum algorithm starts with many
copies of the state $\left\vert S\right\rangle $. \ As a consequence of
this---and specifically, of the ${1}/\sqrt{\left\vert S\right\vert }$\
normalizing factor in $\left\vert S\right\rangle $---when we write the
average acceptance probability of our algorithm as a function of $\left\vert
S\right\vert $, we find that we get a \textit{Laurent polynomial}: a
polynomial that can contain both positive and negative integer powers of $%
\left\vert S\right\vert $. \ The degree of this polynomial (the highest
power of $\left\vert S\right\vert $) encodes the sum of the number of
queries, the number of copies of $\left\vert S\right\rangle $, and the
number of uses of $\mathcal{R}_S$, while the \textquotedblleft
anti-degree\textquotedblright\ (the highest power of ${\left\vert
S\right\vert^{-1} }$) encodes the sum of the number of copies of $\left\vert
S\right\rangle $ and number of uses of $\mathcal{R}_S$. 
This is described more precisely in \Cref{sec:Laurent}.
We're thus faced with the
task of lower-bounding the degree and the anti-degree of a Laurent
polynomial that's bounded in $[0,1]$ at integer points and that encodes the
approximate counting problem.

We then lower bound the degree of Laurent polynomials that approximate $%
\mathsf{ApxCount}_{N,w}$, showing that degree $\Omega \bigl(\min \bigl\{%
w^{1/3},\sqrt{{N}/{w}}\bigr\}\bigr)$ is necessary. \ We give two very different lower bound arguments. \ The first approach, which we call the
\textquotedblleft explosion argument,\textquotedblright\ is shorter but
yields suboptimal lower bounds, whereas the second approach using
\textquotedblleft dual polynomials\textquotedblright\ yields the optimal
lower bound.

There are two aspects of this that we find surprising: first, that Laurent
polynomials appear at all, and second, that they seem to appear in a
completely different way than they appear in our other result about $\mathsf{QMA}$ (\Cref{thm:sbqp}),
despite the close connection between the two statements. \ 
For \Cref{thm:main}, Laurent polynomials are needed just to describe the quantum
algorithm's acceptance probability, whereas for \Cref{thm:sbqp}, ordinary (bivariate)
polynomials sufficed to describe this probability; Laurent polynomials appeared only when we restricted a
bivariate polynomial to a hyperbola in the plane. \ 
In any case, the coincidence suggests that the \textquotedblleft Laurent polynomial
method\textquotedblright\ might be useful for other problems as well.\footnote{Since writing this, a third application of the Laurent polynomial method was discovered by the third author \cite{kretschmer}: a simple proof that the $\mathsf{AND}$-$\mathsf{OR}$ tree $\mathsf{AND}_m \circ \mathsf{OR}_n$ has approximate degree $\widetilde{\Omega}(\sqrt{mn})$.}

Before describing our techniques at a high level, observe that there are
\emph{rational} functions\footnote{A rational function of degree $d$ is of the form $\frac{p(x)}{q(x)}$, where $p$ and $q$ are both real polynomials of degree at most $d$.} of degree $O(\log (N/w))$ that approximate $%
\mathsf{ApxCount}_{N,w}$. \ This follows, for example, from Aaronson's $%
\mathsf{PostBQP}=\mathsf{PP}$ theorem \cite{aar:pp}, or alternately from
the classical result of Newman \cite{newman} that for any $k>0$, there is a
rational polynomial of degree $O(k)$ that pointwise approximates the sign
function on domain $[-n,-1]\cup \lbrack 1,n]$ to error $1-n^{-1/k}$. \ Thus,
our proof relies on the fact that Laurent polynomials are an extremely
special kind of rational function. 

We also remark that in the randomized classical setting, the complexity of $\ApxCount$ with queries and uniform (classical) samples is easily characterized without such powerful techniques. \ 
Either $O(N/w)$ queries or $O(\sqrt{w})$ samples are sufficient, and furthermore either $\Omega(N/w)$ queries or $\Omega(\sqrt{w})$ samples are necessary. \ 
For completeness, we provide a sketch of these bounds in \Cref{sec:classical_samples_queries}.

\para{Overview of the explosion argument.}
Our first proof (in \Cref{sec:explosion}) uses 
an \textquotedblleft explosion argument\textquotedblright\
that, as far as we know, is new in quantum query complexity. \ We separate
out the purely positive degree\footnote{%
Throughout this paper we allow any \textquotedblleft purely positive
degree\textquotedblright\ Laurent polynomial and any \textquotedblleft
purely negative degree\textquotedblright\ Laurent polynomial to include a
constant (degree zero) term.} and purely negative degree parts of our
Laurent polynomial as $q\left( \left\vert S\right\vert \right) =u\left(
\left\vert S\right\vert \right) +v({1}/{\left\vert S\right\vert })$, where $%
u $ and $v$ are ordinary polynomials. \ We then show that, if $u$ and $v$
both have low enough degree, namely $\deg \left( u\right) =o\bigl(\sqrt{{N}/{%
w}}\bigr)$ and $\deg \left( v\right) =o\left( w^{1/4}\right) $, then we get
\textquotedblleft unbounded growth\textquotedblright\ in their values. \
That is: for approximation theory reasons, either $u$ or $v$ must attain
large values, far outside of $\left[ 0,1\right] $, at some integer values of
$\left\vert S\right\vert $. \ But that means that, for $q$ itself to be
bounded in $\left[ 0,1\right] $\ (and thus represent a probability), the
other polynomial must \textit{also} attain large values. \ And that, in
turn, will force the first polynomial to attain even larger values, and so
on forever---thereby proving that these polynomials could not have existed.

\para{Overview of the method of dual polynomials.} 
Our second argument (in \Cref{sec:dualpoly}) obtains the (optimal) lower bound stated
in \Cref{thm:main}, via a novel adaptation of the so-called \emph{method of
dual polynomials}.

With this method, to lower-bound the approximate degree of a Boolean
function $f$, one exhibits an explicit \emph{dual polynomial} $\psi $ for $f$%
, which is a dual solution to a certain linear program. Roughly speaking, a
dual polynomial $\psi $ is a function mapping the domain of $f$ to $\mathbb{R%
}$ that is (a) uncorrelated with any polynomial of degree at most $d$, and
(b) well-correlated with $f$.

Approximating a univariate function $g$ via low-degree Laurent polynomials
is also captured by a linear program, but the linear program is more
complicated because Laurent polynomials can have negative-degree terms. \ We
analyze the value of this linear program in two steps.

In Step 1, we transform the linear program so that it refers only to
ordinary polynomials rather than Laurent polynomials. \ Although simple,
this transformation is crucial, as it lets us bring techniques developed for
ordinary polynomials to bear on our goal of proving Laurent polynomial
degree lower bounds.

In Step 2, we explicitly construct an optimal dual witness to the
transformed linear program from Step 1. \ We\ do so by first identifying two
weaker dual witnesses: $\psi _{1}$, which witnesses that \emph{ordinary} (i.e., purely positive degree)
polynomials encoding approximate counting require degree at least $\Omega
\bigl( \sqrt{N/w}\bigr) $, and $\psi _{2}$, which witnesses that
 purely negative degree polynomials
encoding approximate counting require degree $\Omega (w^{1/3})$. \ The first
witness is derived from prior work of Bun and Thaler \cite{bt13} (who refined earlier work of {\v{S}}palek \cite{spalek}), while the
second builds on a non-constructive argument of Zhandry \cite{zhandry}.

Finally, we show how to \textquotedblleft glue together\textquotedblright\ $%
\psi _{1}$ and $\psi _{2}$, to get a dual witness $\psi $ showing that any
general Laurent polynomial that encodes approximate counting must have
either positive degree $\Omega \bigl( \sqrt{N/w}\bigr) $ or negative degree
$\Omega (w^{1/3})$.

\para{Overview of the upper bound.}To recap, %
\Cref{thm:main} shows that any quantum algorithm for $\mathsf{ApxCount}%
_{N,w} $ needs either $\Theta (\sqrt{N/w})$ queries or $\Theta \bigl(\min %
\bigl\{w^{1/3},\sqrt{{N}/{w}}\bigr\}\bigr)$ samples and reflections. \ Since
we know from the work of Brassard, H{\o }yer, Tapp~\cite{bht:count} that the
problem can be solved with $O(\sqrt{N/w})$ queries alone, it remains only to
show the matching upper bound using samples and reflections, which we 
describe in \Cref{sec:upper}.

First we describe a simple algorithm that uses $O(\sqrt{N/w})$ samples and
reflections. If we take one copy of $\left\vert S\right\rangle $, and
perform a projective measurement onto $|\psi \>=\frac{1}{\sqrt{N}}\left(
\left\vert 1\right\rangle +\cdots +\left\vert N\right\rangle \right) $ or
its orthogonal complement, the measurement will succeed with probability $|%
\langle S|\psi \>|^{2}=\left\vert S\right\vert /N$. \ Thus $O(N/w)$ repetitions of this will allow us to distinguish the probabilities $w/N$
and $2w/N$. We can improve this by using amplitude amplification~\cite{BHMT02} and only make $O(\sqrt{N/w})$ repetitions. \


Our second algorithm solves the problem with $O(w^{1/3})$ reflections and
samples and is based on the quantum collision-finding algorithm~\cite{BHT98}%
. \ We first use $O(w^{1/3})$ copies of $|S\>$ to learn $w^{1/3}$ distinct
elements in $S$. \ We now know a fraction of elements in $S$, and this
fraction is either $w^{-2/3}$ or $\frac{1}{2}w^{-2/3}$. \ We then use
amplitude amplification (or quantum counting) to distinguish these two cases,
which costs $O(w^{1/3})$ repetitions, where each repetition uses a
reflection about $|S\>$.

\section{Preliminaries}

\label{sec:prelim}

In this section we introduce some definitions and known facts about
polynomials and complexity classes.

\subsection{Approximation theory}
\label{s:approxprelim}
We will use several results from approximation theory,\ each of which has
previously been used (in some form) in other applications of the polynomial
method to quantum lower bounds. \ We start with the basic inequality of A.A.
Markov \cite{markov1890question}.

\begin{lemma}[Markov]
\label{markovlem}Let $p$\ be a real polynomial, and suppose that%
\begin{equation}
\max_{x,y\in\left[ a,b\right] }\left\vert p\left( x\right) -p\left( y\right)
\right\vert \leq H.
\end{equation}
Then for all $x\in\left[ a,b\right] $, we have
\begin{equation}
\left\vert p^{\prime}\left( x\right) \right\vert \leq\frac{H}{b-a}\deg\left(
p\right) ^{2},
\end{equation}
where $p^{\prime}(x)$ is the derivative of $p$ at $x$.
\end{lemma}

We'll also need a bound that was explicitly stated by Paturi \cite{paturi},
and which amounts to the fact that, among all degree-$d$
polynomials that are bounded within a given range, the Chebyshev polynomials
have the fastest growth outside that range.

\begin{lemma}[Paturi]
\label{paturilem}Let $p$\ be a real polynomial, and suppose that $\left\vert
p\left( x\right) \right\vert \leq1$\ for all $\left\vert x\right\vert \leq1$%
. \ Then for all $x\leq1+\mu$, we have%
\begin{equation}
\left\vert p\left( x\right) \right\vert \leq\exp\left( 2\deg\left( p\right)
\sqrt{2\mu+\mu^{2}}\right) .
\end{equation}
\end{lemma}

We now state a useful corollary of \Cref{paturilem}, which says (in effect)
that slightly shrinking the domain of a low-degree real polynomial can only
modestly shrink its range.

\begin{corollary}
\label{paturicor}Let $p$\ be a real polynomial of degree $d$, and suppose
that%
\begin{equation}
\max_{x,y\in\left[ a,b\right] }\left\vert p\left( x\right) -p\left( y\right)
\right\vert \geq H.
\end{equation}
Let $\varepsilon\leq\frac{1}{100d^{2}}$\ and $a^{\prime}:=a+\varepsilon%
\left( b-a\right) $. \ Then%
\begin{equation}
\max_{x,y\in\left[ a^{\prime},b\right] }\left\vert p\left( x\right) -p\left(
y\right) \right\vert \geq\frac{H}{2}.
\end{equation}
\end{corollary}

\begin{proof}
Suppose by contradiction that%
\begin{equation}
\left\vert p\left(  x\right)  -p\left(  y\right)  \right\vert <\frac{H}{2}%
\end{equation}
for all $x,y\in\left[  a^{\prime},b\right]  $. \ By affine shifts, we can
assume without loss of generality that $\left\vert p\left(  x\right)
\right\vert <\frac{H}{4}$\ for all $x\in\left[  a^{\prime},b\right]  $. \ Then
by \Cref{paturilem}, for all $x\in\left[  a,b\right]  $\ we have%
\begin{equation}
\left\vert p\left(  x\right)  \right\vert <\frac{H}{4}\cdot\exp\left(
2d\sqrt{2\left(  \frac{1}{1-\varepsilon}-1\right)  +\left(  \frac
{1}{1-\varepsilon}-1\right)  ^{2}}\right)  \leq\frac{H}{2}.
\end{equation}
But this violates the hypothesis.
\end{proof}

We will also need a bound that relates the range of a low-degree polynomial
on a discrete set of points to its range on a continuous interval. \ The
following lemma generalizes a result due to Ehlich and Zeller \cite{ez} and
Rivlin and Cheney \cite{rc}, who were interested only in the case where the
discrete points are evenly spaced.

\begin{lemma}
\label{ezrclem}Let $p$\ be a real polynomial of degree at most $\sqrt{k}$,
and let $0=z_{1}<\cdots<z_{M}=k$\ be a list of points such that $%
z_{i+1}-z_{i}\leq1$\ for all $i$ (the simplest example being the integers $%
0,\ldots,k$). \ Suppose that%
\begin{equation}
\max_{x,y\in\left[ 0,k\right] }\left\vert p\left( x\right) -p\left( y\right)
\right\vert \geq H.
\end{equation}
Then%
\begin{equation}
\max_{i,j}\left\vert p\left( z_{i}\right) -p\left( z_{j}\right) \right\vert
\geq\frac{H}{2}.
\end{equation}
\end{lemma}

\begin{proof}
Suppose by contradiction that%
\begin{equation}
\left\vert p\left(  z_{i}\right)  -p\left(  z_{j}\right)  \right\vert
<\frac{H}{2}%
\end{equation}
for all $i,j$. \ By affine shifts, we can assume without loss of generality
that $\left\vert p\left(  z_{i}\right)  \right\vert <\frac{H}{4}$\ for all
$i$. \ Let%
\begin{equation}
c:=\max_{x\in\left[  0,k\right]  }\frac{\left\vert p\left(  x\right)
\right\vert }{H/4}.
\end{equation}
If $c\leq1$, then the hypothesis clearly fails, so assume $c>1$. \ Suppose
that the maximum, $\left\vert p\left(  x\right)  \right\vert =\frac{cH}{4}$,
is achieved between $z_{i}$\ and $z_{i+1}$. \ Then by basic calculus, there
exists an $x^{\ast}\in\left[  z_{i},z_{i+1}\right]  $\ such that%
\begin{equation}
\left\vert p^{\prime}\left(  x^{\ast}\right)  \right\vert >\frac{2\left(
c-1\right)  }{z_{i+1}-z_{i}}\cdot\frac{H}{4}\geq\frac{\left(  c-1\right)
H}{2}.
\end{equation}
So by \Cref{markovlem},%
\begin{equation}
\frac{\left(  c-1\right)  H}{2}<\frac{cH/4}{k}\deg\left(  p\right)  ^{2}.
\end{equation}
Solving for $c$, we find%
\begin{equation}
c<\frac{2k}{2k-\deg\left(  p\right)  ^{2}}\leq2.
\end{equation}
But if $c<2$, then $\max_{x\in\left[  0,k\right]  }\left\vert p\left(
x\right)  \right\vert <\frac{H}{2}$, which violates the hypothesis.
\end{proof}

We also use a related inequality due to Coppersmith and Rivlin~\cite%
{coppersmith-rivlin} that bounds a polynomial on a continuous interval in
terms of a bound on a discrete set of points, but now with the weaker
assumption that the degree is at most $k$, rather than $\sqrt{k}$. \ This
gives a substantially weaker bound.

\begin{lemma}[Coppersmith and Rivlin]
\label{lem:coppersmith_rivlin} Let $p$ be a real polynomial of degree at
most $k$, and suppose that $\left\vert p(x)\right\vert \leq 1$ for all
integers $x\in \{0,1,\ldots ,k\}$. \ Then there exist universal constants $%
a,b$ such that for all $x\in \lbrack 0,k]$, we have
\begin{equation}
\left\vert p(x)\right\vert \leq a\cdot \exp \left( b\,\deg (p)^{2}/k\right) .
\end{equation}
\end{lemma}

\subsection{Symmetric polynomials}

\para{Univariate symmetrizations.}
Our starting point is the well-known \textit{symmetrization lemma} of Minsky
and Papert \cite{mp} (see also Beals et al.\ \cite{bbcmw}\ for its
application to quantum query complexity), by which we can often reduce
questions about multivariate polynomials to questions about univariate ones.

\begin{lemma}[Minsky--Papert symmetrization]
\label{symlem} Let $p:\left\{ 0,1\right\} ^{N}\rightarrow\mathbb{R}$\ be a
real multilinear polynomial of degree $d$, and let $q:\{0,1,\ldots,N\}\to
\mathbb{R}$ be defined as
\begin{equation}
q\left( k\right) :=\E_{\left\vert X\right\vert =k}\left[ p\left( X\right) %
\right] .
\end{equation}
Then $q$ can be written as a real polynomial in $k$ of degree at most $d$.
\end{lemma}

We now introduce a different, lesser known
notion of symmetrization, 
which we call the \emph{erase-all-subscripts} symmetrization
for reasons to be explained shortly. This symmetrization previously appeared in \cite{shi} under the name ``linearization,'' and it is also equivalent to the noise operator used in analysis of Boolean functions \cite[Definition 2.46]{odonnell}.

\begin{lemma}[Erase-all-subscripts
symmetrization]
Let $p:\left\{ 0,1\right\} ^{N}\rightarrow\mathbb{R}$\ be a
real multilinear polynomial of degree $d$,
and for
any real number $k \in [0, 1]$,
let $M_{k}$ denote
the distribution over $\{0, 1\}^N$, wherein each coordinate is selected
independently to be 1 with probability $k$.
Let $q:[0, 1] \to
\mathbb{R}$ be defined as
\begin{equation}
q\left( k\right) :=\E_{X \sim M_{k}}\left[ p\left( X\right) %
\right] .
\end{equation}
Then $q$ can be written as a real polynomial in $k$ of degree at most $d$.
\label{lem:eraseallsubscripts}
\end{lemma}
\begin{proof} (see, for example, \cite[Proof of Theorem 3]{STT12}).
Given the multivariate polynomial expansion of $p$,
we can obtain $q$ easily just by ``erasing all the subscripts in each variable''. For example, if
$p(x_1, x_2, x_3)= 2 x_1 x_2 + x_2 x_3 + x_2$, we replace every $x_i$ with $k$ to obtain
$q(k)=2k\cdot k+ k\cdot k + k=3k^2 + k$.
This follows from linearity of expectation along with the fact that $M_k$
is defined to be the
product distribution wherein
each coordinate has expected value $k$.
\end{proof}

We highlight the following key difference between Minsky--Papert symmetrization
and the erase-all-subscripts symmetrization. Let $p:\left\{ 0,1\right\} ^{N}\rightarrow [0, 1]$
be a real multivariate polynomial whose evaluations at Boolean inputs are in
$[0, 1]$, i.e., for all $x\in\{0,1\}^n$, we have $p(x)\in[0,1]$.
If $q$ is the erase-all-subscripts symmetrization of $p$, 
then $q$ takes values in $[0,1]$ at all \emph{real-valued}
inputs in $[0, 1]$: $q(k) \in [0,1]$ for all $k \in [0,1]$. 
If $q$ is the Minsky--Papert symmetrization of $p$,
then it is only guaranteed to take values in $[0,1]$ at 
\emph{integer-valued} inputs in $[0, N]$, i.e., $q(k)\in [0,1]$ 
is only guaranteed to hold at $k\in\{0,1,\ldots,N\}$.
This is the main reason we use erase-all-subscripts symmetrization
in this work.

\para{Bivariate symmetrizations.}
In this paper, it will be convenient
to consider bivariate versions of both
Minsky--Papert and erase-all-subscripts
symmetrization, and their applications to oracle separations. To this end,
define $X\in \left\{ 0,1\right\} ^{N}$, the \textquotedblleft characteristic
string\textquotedblright\ of the set $S\subseteq \left[ N\right] $, by $%
x_{i}=1$\ if $i\in S$\ and $x_{i}=0$\ otherwise. \ Let $\mathcal{O}_{S}$
denote the unitary that performs a membership query to $S$, defined as
\begin{equation}
\mathcal{O}_{S}\left\vert i\right\rangle \left\vert b\right\rangle
=(1-2bx_{i})\left\vert i\right\rangle \left\vert b\right\rangle
\end{equation}%
for any index $i\in \lbrack N]$ and bit $b\in \{0,1\}$.

Because we study oracle intersection problems, it is often convenient to
think of an algorithm as having access to \textit{two} oracles, wherein the
first bit in the oracle register selects the choice of oracle. \ As a
consequence, we need a slight generalization of a now well-established fact
in quantum complexity: that the acceptance probability of a quantum
algorithm with an oracle can be expressed as a polynomial in the bits of the
oracle string.

\begin{lemma}[Symmetrization with two oracles]
\label{lem:symmetrization} Let $Q^{\mathcal{O}_{S_{0}},\mathcal{O}_{S_{1}}}$
be a quantum algorithm that makes $T$ queries to a pair of membership
oracles for sets $S_{0},S_{1}\subseteq \lbrack N]$. \ Let $D_{\mu }$ denote
the distribution over subsets of $[N]$ wherein each element is selected
independently with probability $\frac{\mu }{N}$. \ Then there exist
bivariate real polynomials $q(s,t)$ and $p(x,y)$ of degree at most $2T$
satisfying:
\begin{align*}
\textrm{for all real numbers } s, t \in [0, N],& \quad
q(s,t)=\E_{\substack{ S_{0}\sim D_{s}, \\ S_{1}\sim D_{t}}}\left[ \Pr [Q^{%
\mathcal{O}_{S_{0}},\mathcal{O}_{S_{1}}}\text{ accepts}]\right], \textrm{ and} \\
\textrm{for all integers } x, y \in \{0, 1, \dots, N\},& \quad
p(x,y)=\E_{\substack{ |S_{0}|=x, \\ |S_{1}|=y}}\left[ \Pr [Q^{\mathcal{O}%
_{S_{0}},\mathcal{O}_{S_{1}}}\text{ accepts}]\right].
\end{align*}
\end{lemma}

\begin{proof}
Take $X = X_0|X_1$ to be the concatenation of the characteristic strings of the two oracles, and let $S \subseteq [2N]$ be such that $X$ is the characteristic string of $S$. Then, Lemma 4.2 of Beals et al.\ \cite{bbcmw} tells us that there is a real multilinear polynomial $r(X)$ of degree at most $2T$ in the bits of $X$ such that $r(X) = \Pr[ Q^{\mathcal{O}_S} \text{ accepts}]$.

Observe that $r$ has a meaningful probabilistic interpretation over arbitrary inputs in $[0, 1]$. A vector $X \in [0, 1]^{2N}$ of probabilities corresponds to a distribution over $\{0,1\}^{2N}$ wherein each bit is chosen from a Bernoulli distribution with the corresponding probability. Because $r$ is multilinear, $r$ in fact computes the expectation of the acceptance probability over this distribution. In particular, the polynomial
\begin{equation}
q(s, t) = r\bigg(\underbrace{\frac{s}{N},\ldots,\frac{s}{N}}_{N \text{ times}}, \underbrace{\frac{t}{N},\ldots,\frac{t}{N}}_{N \text{ times}}\bigg) = \E_{\substack{S_0 \sim D_s,\\S_1 \sim D_t}}\left[\Pr[ Q^{\mathcal{O}_{S_0},\mathcal{O}_{S_1}} \text{ accepts}] \right]    
\end{equation}
corresponds to selecting $S_0 \sim D_s$ and $S_1 \sim D_t$. The total degree of $q$ is obviously
at most the degree of $r$, by the same reasoning as in the proof of \Cref{lem:eraseallsubscripts}.

To construct $p$, we apply the symmetrization lemma of Minsky and Papert \cite{mp} to symmetrize $r$, first with respect to $X_0$, then with respect to $X_1$:
\begin{equation}
    p_0(x, X_1) = \E_{|S_0|=x} r(X_0,X_1) =  \E_{|S_0|=x}\left[\Pr[ Q^{\mathcal{O}_{S_0},\mathcal{O}_{S_1}} \text{ accepts}] \right]
\end{equation}
\begin{equation}
    p(x, y) = \E_{|S_1|=y} p_0(x,X_1) = \E_{\substack{|S_0|=x,\\|S_1|=y}}\left[\Pr[ Q^{\mathcal{O}_{S_0},\mathcal{O}_{S_1}} \text{ accepts}] \right] 
\end{equation}  
The degree of $p$ is at most the degree of $r$, due to \Cref{symlem}.
\end{proof}

We remark that, as a consequence of their definitions in \Cref{lem:symmetrization}, $p$ and $q$ satisfy:
\begin{equation}
q(s, t) = \E \left[ p(X, Y) \right],
\end{equation}
where $X$ and $Y$ are drawn from $N$-trial binomial distributions with means $s$ and $t$, respectively.

\para{Symmetric Laurent polynomials.}
Finally, we state a useful fact about Laurent polynomials:

\begin{lemma}[Symmetric Laurent polynomials]
\label{lem:symmetric} Let $\ell (x)$ be a real Laurent polynomial of positive
and negative degree $d$ that satisfies $\ell (x)=\ell (1/x)$. \ Then there exists a (ordinary) real
polynomial $q$ of degree $d$ such that $\ell (x)=q(x+1/x)$.
\end{lemma}

\begin{proof}
$\ell(x) = \ell(1/x)$ implies that the coefficients of the $x^i$ and $x^{-i}$ terms are equal for all $i$, as otherwise $\ell(x) - \ell(1/x)$ would not equal the zero polynomial. Thus, we may write $\ell(x) = \sum_{i=0}^d a_i \cdot (x^i + x^{-i})$ for some coefficients $a_i$. So, it suffices to show that $x^i + x^{-i}$ can be expressed as a polynomial in $x + 1/x$ for all $0 \le i \le d$.

We prove by induction on $i$. The case $i = 0$ corresponds to constant polynomials. For $i > 0$, by the binomial theorem, observe that $(x + 1/x)^i = x^i + x^{-i} + r(x)$ where $r$ is a degree $i - 1$ real Laurent polynomial satisfying $r(x) = r(1/x)$. By the induction assumption, $r$ can be expressed as a polynomial in $x + 1/x$, so we have $x^i + x^{-i} = (x + 1/x)^i - r(x)$ is expressed as a polynomial in $x + 1/x$.
\end{proof}

\subsection{Complexity classes}

\begin{definition}
\label{def:qma}
The complexity class $\mathsf{QMA}$ consists of the
languages $L$ for which there exists a quantum polynomial time
verifier $V$ with the following properties:

\begin{enumerate}
\item Completeness: if $x \in L$, then there exists a quantum witness state $|\psi\>$ on $\mathrm{poly}(|x|)$ qubits such that $\Pr\left[V(x, |\psi\>) \text{ accepts}\right] \ge \frac{2}{3}$.

\item Soundness: if $x \not\in L$, then for any quantum witness state $|\psi\>$ on $\mathrm{poly}(|x|)$ qubits, $\Pr\left[V(x, |\psi\>) \text{ accepts}\right] \le \frac{1}{3}$.
\end{enumerate}
\end{definition}

A quantum verifier that satisfies the above
promise for a particular language will be referred to as a $\mathsf{QMA}$
verifier or $\mathsf{QMA}$ protocol throughout.

Though $\mathsf{SBP}$ and $\mathsf{SBQP}$ can be defined in terms of
counting complexity functions, for our purposes it is easier to
work with the following equivalent definitions (see B\"ohler et al.\ \cite%
{BGM06}):

\begin{definition}
\label{def:sbp_sbqp} The complexity class $\mathsf{SBP}$ consists of the
languages $L$ for which there exists a probabilistic polynomial time
algorithm $M$ and a polynomial $\sigma$ with the following properties:

\begin{enumerate}
\item If $x \in L$, then $\Pr\left[M(x) \text{ accepts}\right] \ge
2^{-\sigma(|x|)}$.

\item If $x \not\in L$, then $\Pr\left[M(x) \text{ accepts}\right] \le
2^{-\sigma(|x|)}/2$.
\end{enumerate}

The complexity class $\mathsf{SBQP}$ is defined analogously, wherein the
classical algorithm is replaced with a quantum algorithm.
\end{definition}

A classical (respectively, quantum) algorithm that satisfies the above
promise for a particular language will be referred to as an $\mathsf{SBP}$
(respectively, $\mathsf{SBQP}$) algorithm throughout. \ Using these
definitions, a  query complexity relation between $\mathsf{QMA}$
protocols and $\mathsf{SBQP}$ algorithms follows from the procedure of
Marriott and Watrous \cite{marriott}, which shows that one can exponentially
improve the soundness and completeness errors of a $\mathsf{QMA}$ protocol
without increasing the witness size. \ 
This relationship is now standard; see
for example \cite[Remark 6]{marriott} or \cite[Proposition 4.2]%
{sherstov2019vanishing} for a proof of the following lemma:

\begin{lemma} 
\label{lem:guessing} Suppose there is a $\mathsf{QMA}$ protocol for some
problem that makes $T$ queries and receives an $m$-qubit witness. \ Then
there is a quantum query algorithm $Q$ for the same problem that makes $O(mT)$ queries, and satisfies the following:

\begin{enumerate}
\item If $x \in L$, then $\Pr\left[Q(x) \text{ accepts}\right] \ge
2^{-m}$.

\item If $x \not\in L$, then $\Pr\left[Q(x) \text{ accepts}\right] \le
2^{-10m}$.
\end{enumerate}
\end{lemma}

\section{QMA complexity of approximate counting}

\label{sec:SBPQMA}

This section establishes
an 
optimal lower bound on the $\mathsf{QMA}$
complexity of approximate counting.
We first lower bound the $\mathsf{SBQP}$ complexity
of the $\mathsf{AND}_2 \circ \mathsf{ApxCount}_{N,w}$ problem (%
\Cref{thm:sbqp}). 
This implies a $\mathsf{QMA}$ lower bound
for $\mathsf{ApxCount}_{N,w}$ via \Cref{lem:guessing},
but it is not quantitatively optimal. 
We prove the optimal $\mathsf{QMA}$ lower bound 
(\Cref{thm:qmabound}) via
\Cref{lem:betterforqma}, which
leverages additional properties 
of 
the $\mathsf{SBQP}$ protocol derived via
\Cref{lem:guessing}
from any $\mathsf{QMA}$ protocol with small
witness length. Finally, \Cref{cor:qma_separation}
describes new oracle separations that
are immediate consequences of \Cref{thm:qmabound} and \Cref{thm:sbqp}. 

\subsection{Lower bound for \texorpdfstring{$\mathsf{SBQP}$}{SBQP} algorithms}
\label{sec:SBQP} 

Our lower bound on the $\cl{SBQP}$ complexity of $\AndApxCount$ hinges 
on the following theorem. \ The theorem
uses Laurent polynomials to prove a degree lower bound for bivariate
polynomials that satisfy an upper bound on an \textquotedblleft
L\textquotedblright -shaped pair of rectangles and a lower bound at a nearby
point:

\begin{theorem}
\label{thm:L} Let $0<w<32w<N$ and $M\geq 1$. \ Let $R_{1}=[4w,N]\times \lbrack 0,w/2]$ and
$R_{2}=[0,w/2]\times \lbrack 4w,N]$ be disjoint rectangles in the plane, and
let $L=R_{1}\cup R_{2}$. \ Let $p(x,y)$ be a real polynomial of degree $d$
with the following properties:
\begin{enumerate}
\item $p(4w, 4w) \ge 1.5 \cdot M$.
\item $0 \le p(x, y) \le 1$ for all $(x, y) \in L$.
\end{enumerate}
Then $d = \Omega(\sqrt{N/w} \cdot \log M)$.
\end{theorem}

\begin{figure}[tbp]
\centering
\begin{tikzpicture}[y=0.4cm, x=0.4cm]
	\pgfmathsetmacro{\N}{20};
	\pgfmathsetmacro{\w}{3};
	\pgfmathsetmacro{\tw}{\w+\w};
	\pgfmathsetmacro{\Np}{\N+1};

	\draw[<-,thick] (\tw+0.5, \w/2) -- (\w+1,\w/2) node[anchor=east] {$R_1$};
	\draw[<-,thick] (\w/2, \tw+0.5) -- (\w/2,\w+1) node[anchor=north] {$R_2$};

	\fill[gray,opacity=.4] (\tw,0) -- (\tw,\w) -- (\N,\w) -- (\N,0);
	\fill[gray,opacity=.4] (0,\tw) -- (\w,\tw) -- (\w,\N) -- (0,\N);

	\draw[->] (0,0) -- coordinate (x axis mid) (\Np,0);
    	\draw[->] (0,0) -- coordinate (y axis mid) (0,\Np);
    	
	\draw[<-,thick] (\tw,\tw) -- (\tw+2,\tw+2) node[anchor=south west] {$t=1$};   	
	\draw[<-,thick] (\w*4,\w) -- (\w*4,\w*2) node[anchor=south] {$t=8$};
	\draw[<-,thick] (\N,4*\w*\w/\N) -- (\N+2,1+4*\w*\w/\N) node[anchor=west] {$t=\frac{N}{4w}$};
	\draw[<-,thick] (1.5*\w,2*\w*4/3) -- (1.5*\w+2,2*\w*4/3+2) node[anchor=south west] {$(x=4wt,y=4w/t)$};
    	
    \draw (0,1pt) -- (0,-3pt) node[anchor=north] {0};
    \draw (\w,1pt) -- (\w,-3pt) node[anchor=north] {$w/2$\vphantom{l}};
    \draw (\tw,1pt) -- (\tw,-3pt) node[anchor=north] {$4w$};
    \draw (\N,1pt) -- (\N,-3pt) node[anchor=north] {$N$};

    \draw (1pt,0) -- (-3pt,0) node[anchor=east] {0};
    \draw (1pt,\w) -- (-3pt,\w) node[anchor=east] {$w/2$};
    \draw (1pt,\tw) -- (-3pt,\tw) node[anchor=east] {$4w$};
    \draw (1pt,\N) -- (-3pt,\N) node[anchor=east] {$N$};

	\node[below=0.8cm] at (x axis mid) {$x$};
	\node[rotate=90, above=0.8cm] at (y axis mid) {$y$};
	\node at (\w+\N/2, \w/2) {$0 \le p(x, y) \le 1$};
	\node[rotate=90] at (\w/2, \w+\N/2) {$0 \le p(x, y) \le 1$};

     \draw[thick,domain=1:{\N/(2*\w)},smooth,variable=\t,blue] plot ({2*\w*\t},{2*\w/\t});
     \draw[thick,domain=1:{\N/(2*\w)},smooth,variable=\t,blue] plot ({2*\w/\t},{2*\w*\t});
     \node at ({\w+\w},{\w+\w}) {\tiny\textbullet};
     \node at (4*\w,\w) {\tiny\textbullet};
     \node at (\N,4*\w*\w/\N) {\tiny\textbullet};
\end{tikzpicture}
\caption{Diagram of \Cref{thm:L} (not drawn to scale).}
\label{fig:L}
\end{figure}
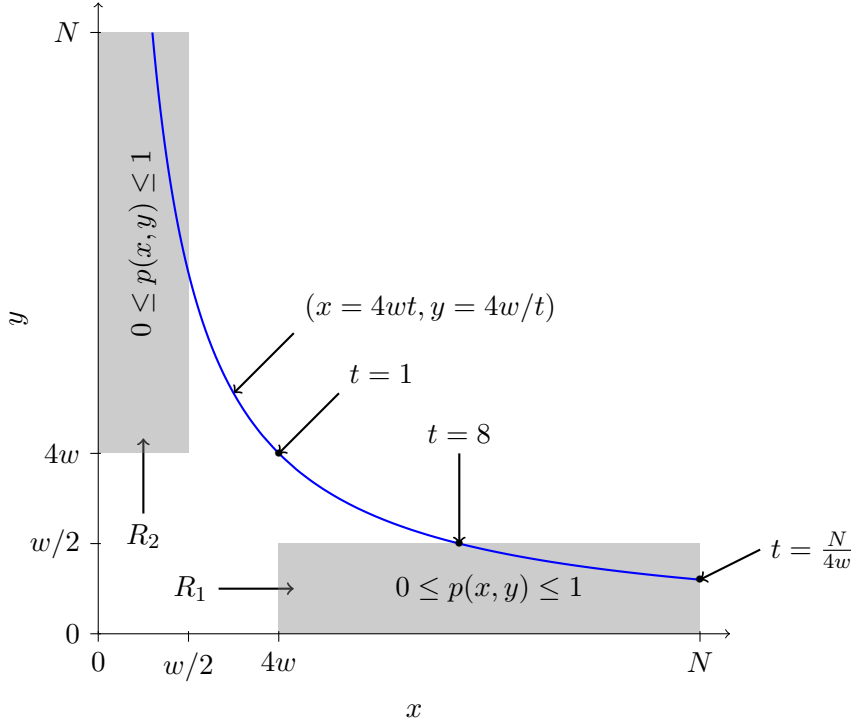

\begin{proof}
Observe that if $p(x, y)$ satisfies the statement of the theorem, then so does $p(y, x)$. This is because the constraints in the statement of the theorem are symmetric in $x$ and $y$ (in particular, because $R_1$ and $R_2$ are mirror images of one another along the line $x = y$; see \Cref{fig:L}). As a result, we may assume without loss of generality that $p$ is symmetric, i.e.,\ $p(x, y) = p(y, x)$. Else, we may replace $p$ by $\frac{p(x, y) + p(y, x)}{2}$ because the set of polynomials that satisfy the inequalities in the statement of the theorem are closed under convex combinations.

Consider the hyperbolic parametric curve $(x = 4wt, y=4w/t)$ as it passes through $R_1$ (see \Cref{fig:L}). We can view the restriction of $p(x, y)$ to this curve as a Laurent polynomial $\ell(t) = p(4wt, 4w/t)$ of positive and negative degree $d$. The bound of $p(x,y)$ on all of $R_1$ implies that $|\ell(t)| \le 1$ when $t \in [8, \frac{N}{4w}]$ and that $\ell(1) \ge 1.5$ (see \Cref{fig:L}). Moreover, the condition that $p(x, y)$ is symmetric implies that $\ell(t) = \ell(1/t)$.

By \Cref{lem:symmetric} for symmetric Laurent polynomials, $\ell(t)$ can be viewed as a degree $d$ polynomial $q(t + 1/t)$. Under the transformation $s = t + 1/t$, $q$ satisfies $|q(s)| \le 1$ for $s \in [8 + 1/8, \frac{N}{4w} + \frac{4w}{N}]$ and $q(2) \ge 1.5 M$. Note that the length of the interval $[8 + 1/8, \frac{N}{4w} + \frac{4w}{N}]$ is $\Theta(N/w)$ because $w < N$. By an appropriate affine transformation of $q$, we can conclude from \Cref{paturilem} with $\mu = \Theta(w/N)$ that $d = \Omega(\sqrt{N/w} \cdot \log M)$.
\end{proof}

Why is \Cref{thm:L} useful? \ One may be tempted to apply this theorem
directly to the polynomial $p(x,y)$ obtained in \Cref{lem:symmetrization} to
conclude a degree lower bound (and thus a query complexity lower bound), as
the \textquotedblleft L\textquotedblright -shaped pair of rectangles $%
L=R_{1}\cup R_{2}$ correspond to \textquotedblleft no\textquotedblright\
instances of $\mathsf{AND}_{2}\circ \mathsf{ApxCount}_{N,w}$, while $(4w,4w)$
corresponds to a \textquotedblleft yes\textquotedblright\ instance. \
However, even though $p(x,y)$ is bounded at lattice points in $L$, it need
not be bounded along the entirety of $L$.\footnote{%
One can nevertheless use this intuition to obtain a nontrivial (though
suboptimal) lower bound by inspecting $p$ alone. Using the Markov brothers'
inequality (\Cref{markovlem}), if $\deg (p)=o(\sqrt{w})$, then the bounds on
$p(x,y)$ at lattice points in $L$ imply that $|p(x,y)|\leq 1+o_{w}(1)$ for
all $(x,y)\in L$. Thus, \Cref{thm:L} applies if $\deg (p)=o(\sqrt{w})$, so
overall we get a lower bound of $\Omega \left( \min \left\{ \sqrt{w},\sqrt{%
N/w}\right\} \right) $ for the $\mathsf{SBQP}$ query complexity of $\mathsf{%
AND}_{2}\circ \mathsf{ApxCount}_{N,w}$. See 
\href{https://arxiv.org/abs/1902.02398}{arXiv:1902.02398} for details.} \ 

To obtain a lower bound, we
instead use the connection between the polynomials $p(x,y)$ and $q(s,t)$
from \Cref{lem:symmetrization}, and establish \Cref{thm:sbqp} from the
introduction, restated for convenience:

\sbqp*

\begin{proof}
Let $N > 32w$ (otherwise the theorem holds trivially). Since $Q$ is an $\mathsf{SBQP}$ algorithm, we may suppose that $Q$ accepts with probability at least $2\alpha$ on a ``yes'' instance and with probability at most $\alpha$ on a ``no'' instance (note that $\alpha$ may be exponentially small in $N$). Take $p(x,y)$ and $q(s, t)$ to be the symmetrized bivariate polynomials of degree at most $2T$ defined in \Cref{lem:symmetrization}. Define $L' = ([0, w] \times [0, w]) \cup ([0, w] \times [2w, N]) \cup ([2w, N] \times [0, w])$. The conditions on the acceptance probability of $Q$ for all $S_0, S_1$ that satisfy the $\ApxCount$ promise imply that $p(x, y)$ satisfies these corresponding conditions:
\begin{enumerate}
\item $1 \ge p(x, y) \ge 2\alpha$ for all $(x, y) \in \left([2w, N] \times [2w, N]\right) \cap \mathbb{Z}^2$.
\item $0 \le p(x, y) \le \alpha$ for all $(x, y) \in L' \cap \mathbb{Z}^2$.
\end{enumerate}

Our strategy is to show that if $T = o(w)$, then these conditions on $p$ imply that the polynomial $q(s, t) \cdot \frac{0.9}{\alpha}$ satisfies the statement of \Cref{thm:L} for all sufficiently large $w$. This in turn implies $T = \Omega(\sqrt{N/w})$. This allows us conclude that either $T = \Omega(w)$ or $T = \Omega(\sqrt{N/w})$, which proves the theorem.

Suppose $T = o(w)$, so that $p(x, y)$ and $q(s, t)$ both have degree $d=o(w)$. We begin by upper bounding $p(x, y)$ at the lattice points $(x, y)$ outside of $L'$. 
We claim the following:
\begin{enumerate}[(a)]
\item $|p(x, y)| \le \alpha \cdot a \cdot \exp(bd^2 / w) \le \alpha \cdot a \cdot \exp(bd)$ whenever $(x, y) \in L'$ and either $x$ or $y$ is an integer, where $a$
and $b$ are the constants from \Cref{lem:coppersmith_rivlin}. This follows from \Cref{lem:coppersmith_rivlin} by fixing either $x$ or $y$ to be an integer and viewing the resulting restriction of $p(x, y)$ as a univariate polynomial in the other variable.
\item $|p(x, y)| \le \alpha \cdot a \cdot \exp(bd) \cdot \exp(2\sqrt{3}d) = \alpha \cdot a \cdot \exp((b + 2\sqrt{3})d)$ whenever $x \in [w, 2w]$, $y \in [0, w]$, and $y$ is an integer. This follows \Cref{paturilem}: consider the univariate polynomial $p(\cdot, y)$ on the intervals $[0, w]$ and $[2w, 3w]$, where it is bounded by (a).
\item $|p(x, y)| \le \alpha \cdot a \cdot \exp((b + 2\sqrt{3})d) \cdot a \cdot \exp(bd^2/w) \le \alpha \cdot a^2 \cdot \exp((2b + 2\sqrt{3})d)$ whenever $x \in [w, 2w]$ and $y \in [0, w]$. This follows from \Cref{lem:coppersmith_rivlin}: consider the univariate polynomial $p(x, \cdot)$ on the interval $[0, w]$, where it is bounded at integer points by (b).
\item $|p(x, y)| \le \alpha \cdot a^2 \cdot \exp((2b + 2\sqrt{3})d) \cdot \exp(4dy/w) = \alpha \cdot a^2 \cdot \exp((2b + 2\sqrt{3} + 4y/w)d)$ whenever $x \in [0, N]$, $y \in [w + 1, N]$, and $x$ is an integer. This follows from  \Cref{paturilem}: consider the univariate polynomial $p(x, \cdot)$ on the interval $[0, w]$, where it is bounded by (a) when $x \in [0, w]$ or $x \in [2w, N]$, or bounded by (c) when $x \in [w, 2w]$. 
By an affine shift, this corresponds to applying \Cref{paturilem} with $\mu = 2y/w - 2$, with the observation that $\sqrt{2\mu + \mu^2} < \mu + 2$.
\end{enumerate}

We now use this to upper bound $q(s, t)$ when $s \in [4w, N]$ and $t \in [0, w/2]$. Let $X$ and $Y$ be drawn from $N$-trial binomial distributions with means $s$ and $t$, respectively,
so that $q(s, t) = \E[p(X,Y)]$. Using the above bounds and basic probability, we have 
\begin{align}
0 \leq q(s, t) = \E[p(X,Y)] &\le \alpha \cdot \bigg( \Pr[X \ge 2w, Y \le w] + \Pr[X \le 2w, Y \le w] \cdot a \cdot \exp\left(\left(b+2\sqrt{3}\right)d\right)  \nonumber \\
& \quad + \sum_{y = w+1}^N \Pr[Y = y] \cdot a^2 \cdot \exp\left(\left(2b + 2\sqrt{3} + 4y/w\right)d\right) \bigg)\\
&\le \alpha \cdot \bigg(1 + \Pr[X \le 2w] \cdot a \cdot \exp\left(\left(b+2\sqrt{3}\right)d\right)   \nonumber \\
& \quad + \sum_{y = w+1}^N \Pr[Y \ge y] \cdot a^2 \cdot \exp\left(\left(2b + 2\sqrt{3} + 4y/w\right)d\right)\bigg).
\end{align}
The probabilities above are easily bounded with a Chernoff bound:
\begin{align}
q(s, t)=\E[p(X,Y)] &\le \alpha \cdot \bigg( 1 + a \cdot \exp\left(\left(b+2\sqrt{3}\right)d - w/2\right)  \nonumber\\
& \quad + \sum_{y = w+1}^N a^2 \cdot \exp\left(\left(2b + 2\sqrt{3} + 4y/w\right)d - y/6\right)\bigg).
\end{align}
Because $a$ and $b$ are universal constants from \Cref{lem:coppersmith_rivlin}, when $d = o(w)$, the first exponential term becomes arbitrarily small for all sufficiently large $w$. Moreover, for all sufficiently large $w$, the remaining sum becomes bounded by a geometric sum. For some constant $c$, we have
\begin{align*}
\sum_{y = w+1}^N a^2 \cdot \exp\left(\left(2b + 2\sqrt{3} + 4y/w\right)d - y/6\right) &\le \sum_{y = w+1}^\infty c \cdot \exp\left(-y/12\right)\\
&\le \frac{c}{1 - \exp(-1/12)} \cdot \exp(-w/12)\\
&= o_w(1).
\end{align*}
Thus we conclude that $0 \le q(s, t) \le \alpha \cdot (1 + o_w(1))$ when $s \in [4w, N]$ and $t \in [0, w/2]$ (i.e.,\ $(s, t) \in R_1$ in the statement of \Cref{thm:L}). By symmetry, we can conclude the same bound when $s \in [0, w/2]$ and $t \in [4w, N]$ (i.e.,\ $(s, t) \in R_2$ in the statement of \Cref{thm:L}).

Now, we lower bound $q(4w, 4w)$. Let $X$ and $Y$ be drawn from independent $N$-trial binomial distributions with mean $4w$, so that $q(4w, 4w) = \mathbb{E}\left[p(X,Y)\right]$. Then we have 
\begin{align*}
\mathbb{E}\left[p(X,Y)\right] &\ge 2\alpha \cdot \Pr[X \ge 2w, Y \ge 2w]\\
&\ge 2\alpha \cdot \left(1 - \Pr[X \le 2w] - \Pr[Y \le 2w]\right)\\
&\ge 2\alpha \cdot \left(1 - 2\exp(-w/2)\right)\\
&\ge 2\alpha \cdot (1 - o_w(1))
\end{align*}
We conclude that $q(s, t) \cdot \frac{0.9}{\alpha}$ satisfies the statement of \Cref{thm:L} (with $M=1$) for all sufficiently large $w$.
\end{proof}

We remark that this lower bound is tight, i.e.,\ there exists an $\mathsf{%
SBQP}$ algorithm that makes $O\left( \min \left\{ w,\sqrt{N/w}\right\}
\right) $ queries. \ The $O(\sqrt{N/w})$ upper bound follows from the $%
\mathsf{BQP}$ algorithm of Brassard, H{\o }yer, and Tapp \cite{bht:count}. \
The $O(w)$ upper bound is in fact an $\mathsf{SBP}$ upper bound with the
following algorithmic interpretation: first, guess $w+1$ items randomly from
each of $S_{0}$ and $S_{1}$. \ Then, verify using the membership oracle that
the first $w+1$ items all belong to $S_{0}$ and that the latter $w+1$ items
all belong to $S_{1}$, accepting if and only if this is the case. \ Clearly,
this accepts with nonzero probability if and only if $\left\vert
S_{0}\right\vert \geq w+1$ and $\left\vert S_{1}\right\vert \geq w+1$. 

\subsection{Lower bound for \texorpdfstring{$\mathsf{QMA}$}{QMA}}
\label{sec:QMA} 

In this section, we establish the optimal $\mathsf{QMA}$ lower bound 
(\Cref{thm:qmabound}).
We begin by quantitatively improving the $\mathsf{SBQP}$
lower bound for $\AndApxCount$ of \Cref{thm:sbqp},
under the stronger assumption that the parameter $\alpha$ 
in the $\mathsf{SBQP}$ protocol is not smaller than $2^{-w}$.
(In addition to a stronger
conclusion, this assumption also permits a considerably simpler
analysis than was required to prove \Cref{thm:sbqp}).

\begin{lemma}
\label{lem:betterforqma}
Consider any quantum query algorithm $Q^{\mathcal{O}_{S_{0}},\mathcal{O}_{S_{1}}}$
for $\AndApxCount$
that makes $T$ queries to the membership oracles
$\mathcal{O}_{S_{0}}$ and $\mathcal{O}_{S_{1}}$
for the two instances of $\ApxCount$ and
satisfies the following. For some $m=o(w)$, $\alpha=2^{-m}$, and $M \in [1, \alpha^{-1}]$:
\begin{enumerate}
\item If $x \in L$, then $\Pr\left[Q(x) \text{ accepts}\right] \ge
\alpha$.
\item If $x \not\in L$, then $\Pr\left[Q(x) \text{ accepts}\right] \le
\alpha/(2 M)$.
\end{enumerate}
Then $T = \Omega\left(\sqrt{N/w} \cdot \log M\right)$
\end{lemma}
\begin{proof}
As in the proof of \Cref{thm:sbqp}, define $L' = ([0, w] \times [0, w]) \cup ([0, w] \times [2w, N]) \cup ([2w, N] \times [0, w])$, and
take $p(x,y)$ and $q(s, t)$ to be the symmetrized bivariate polynomials of degree at most $2T$ defined in \Cref{lem:symmetrization}. $p(x, y)$ satisfies the following 
properties. 
\begin{enumerate}
\item[(a)] $1 \ge p(x, y) \ge \alpha$ for all $(x, y) \in \left([2w, N] \times [2w, N]\right) \cap \mathbb{Z}^2$.
\item[(b)] $0 \le p(x, y) \le \alpha/(1.5 M)$ for all $(x, y) \in L' \cap \mathbb{Z}^2$.
\item[(c)] $0 \leq p(x,y) \leq 1$ for all $(x, y) \in \left( [0,N] \times [0, N]\right) \cap \mathbb{Z}^2$.
\end{enumerate}

We use these properties to upper bound $q(s, t)$ when $s \in [4w, N]$ and $t \in [0, w/2]$. Let $X$ and $Y$ be drawn from $N$-trial binomial distributions with means $s$ and $t$, respectively,
so that $q(s, t) = \E[p(X,Y)]$. Using the above bounds and basic probability, we have 
\begin{align*}
0 \leq q(s, t)=\E[p(X,Y)] &\le \alpha/(2M) \Pr[X \ge 2w, Y \le w]  + 
(1-\Pr[X \ge 2w, Y \le w])
\\
&\le  \alpha/(2M) + 2^{-\Omega(w)} \leq (1+o(1)) \alpha/(2 M)
\end{align*}
Here, the first inequality holds by Properties (a)-(c) above,
while the second follows from a Chernoff Bound, and the third holds because
$\alpha/(2 M) \geq 2^{-o(w)}$.

Thus we conclude that $0 \le q(s, t) \le \alpha/(2M) \cdot (1 + o_w(1))$ when $s \in [4w, N]$ and $t \in [0, w/2]$ (i.e.,\ $(s, t) \in R_1$ in the statement of \Cref{thm:L}). By symmetry, we can conclude the same bound when $s \in [0, w/2]$ and $t \in [4w, N]$ (i.e.,\ $(s, t) \in R_2$ in the statement of \Cref{thm:L}).

Now, we lower bound $q(4w, 4w)$. Let $X$ and $Y$ be drawn from independent $N$-trial binomial distributions with mean $4w$, so that $q(4w, 4w) = \mathbb{E}\left[p(X,Y)\right]$. Then we have 
\begin{align*}
\mathbb{E}\left[p(X,Y)\right] &\ge \alpha \cdot \Pr[X \ge 2w, Y \ge 2w]\\
&\ge \alpha \cdot \left(1 - \Pr[X \le 2w] - \Pr[Y \le 2w]\right)\\
&\ge \alpha \cdot \left(1 - 2\exp(-w/2)\right)\\
&\ge \alpha \cdot (1 - o_w(1))
\end{align*}
We conclude that $q(s, t) \cdot \frac{1.8 M}{\alpha}$ satisfies the statement of \Cref{thm:L} for all sufficiently large $w$.
Hence, $T = \Omega\left(\sqrt{N/w} \cdot \log M\right)$ as claimed.
\end{proof}

We now establish \Cref{thm:qmabound} from the introduction,
which quantitatively 
lower bounds the $\mathsf{QMA}$ complexity of $\mathsf{ApxCount}_{N,w}$. 
The analysis exploits
two key properties of the $\mathsf{SBQP}$ protocols that result
from applying \Cref{lem:guessing} to 
a $\mathsf{QMA}$ protocol with witness length $m$: 
(1) the 
parameter $\alpha$ of the $\mathsf{SBQP}$
protocol is not too small (at least $2^{-m}$)
and (2) the multiplicative gap between acceptance probabilities when $f(x)=0$
vs. $f(x)=1$ is at least $2^m$, which may be much greater than $2$.

\qmabound*

\begin{proof}
Consider a $\mathsf{QMA}$
protocol for $\ApxCount$ with witness size $m$ and query cost $T$. If $m=\Omega(w)$, the theorem is vacuous, so suppose
that $m=o(w)$. Running the verifier, Arthur, a constant number of times with fresh witnesses to reduce the soundness and completeness errors, one obtains a verifier with soundness and completeness errors $1/6$ that receives an $O(m)$-length witness and makes $O(T)$ queries. 
Repeating twice with two oracles and computing the $\mathsf{AND}$, one obtains a $\mathsf{QMA}$ verifier $V'^{\mathcal{O}_{S_0},\mathcal{O}_{S_1}}$ for $\AndApxCount$ with soundness and completeness errors $1/3$ that receives an $O(m)$-length witness and makes $O(T)$ queries. 
Applying \Cref{lem:guessing} to $V'$, there exists a quantum query algorithm $Q^{\mathcal{O}_{S_0},\mathcal{O}_{S_1}}$ for $\AndApxCount$ that makes $O(m \cdot T)$ queries and satisfies the hypothesis of \Cref{lem:betterforqma} with $M=2^{-\Theta(m)}$. \Cref{thm:sbqp} tells us that $m \cdot T = \Omega\left(\sqrt{N/w} \cdot m\right)$.
Equivalently, $T = \Omega\left(\sqrt{N/w}\right)$.
\end{proof}

\Cref{thm:sbqp} also implies several oracle separations:

\begin{corollary}
\label{cor:qma_separation} There exists an oracle $A$ and a pair of
languages $L_0, L_1$ such that:

\begin{enumerate}
\item $L_0, L_1 \in \mathsf{SBP}^A$

\item $L_0 \cap L_1 \not\in \mathsf{SBQP}^A$.

\item $\mathsf{SBP}^A \not\subset \mathsf{QMA}^A$.
\end{enumerate}
\end{corollary}

\begin{proof}
For an arbitrary function $A: \{0,1\}^* \to \{0,1\}$ and $i \in \{0,1\}$, define $A_i^n = \{x \in \{0,1\}^n : A(i, x) = 1\}$. Define the unary language $L^A_i = \{1^n : |A_i^n| \ge 2^{n/2}\}$. Observe that as long as $A$ satisfies the promise $|A_i^n| \ge 2^{n/2}$ or $|A_i^n| \le 2^{n/2-1}$ for all $n \in \mathbb{N}$, then $L^A_i \in \mathsf{SBP}^A$. Intuitively, the oracles $A$ that satisfy this promise encode a pair of $\ApxCount$ instances $|A_0^n|$ and $|A_1^n|$ for every $n \in \mathbb{N}$ where $N = 2^n$ and $w = 2^{n/2 - 1}$.

\Cref{thm:sbqp} tells us that an $\mathsf{SBQP}$ algorithm $Q$ that makes $o(2^{n/4})$ queries fails to solve $\AndApxCount$ on \textit{some} pair $(S_0, S_1)$ that satisfies the promise. Thus, one can construct an $A$ such that $L_0, L_1 \in \mathsf{SBP}^A$ and $L_0 \cap L_1 \not\in \mathsf{SBQP}^A$, by choosing $(A_0^n, A_1^n)$ so as to diagonalize against all $\mathsf{SBQP}$ algorithms.

Because $\mathsf{QMA}^A$ is closed under intersection for any oracle $A$, and because $\mathsf{QMA}^A \subseteq \mathsf{SBQP}^A$ for any oracle $A$, it must be the case that either $L_0 \not\in \mathsf{QMA}^A$ or $L_1 \not\in \mathsf{QMA}^A$.
\end{proof}


\section{Approximate counting with quantum samples and reflections}

\label{sec:samplesreflections}

\subsection{The Laurent polynomial method}

\label{RESULT}\label{sec:lower}\label{sec:Laurent}

By using Minsky--Papert symmetrization (\Cref{symlem}), we now prove the key
fact that relates quantum algorithms, of the type we're considering, to real
Laurent polynomials in one variable. \ The following lemma generalizes the
connection between quantum algorithms and real polynomials established by
Beals et al.\ \cite{bbcmw}.

\begin{lemma}
\label{laurentlem} Let $Q$ be a quantum algorithm that makes $T$ queries to $%
\mathcal{O}_{S}$, uses $R_1$ copies of $\left\vert S\right\rangle $, and
makes $R_2$ uses of the unitary $\mathcal{R}_S$. \ Let $R:= R_1 + 2R_2$. 
For $k\in\{1,\ldots,N\}$, let
\begin{equation}
q\left( k\right) :=\E_{\left\vert S\right\vert =k}\left[ \Pr\left[ Q^{%
\mathcal{O}_{S},\mathcal{R}_S}\left( \left\vert S\right\rangle ^{\otimes
R_1}\right) \text{ accepts}\right] \right] .
\end{equation}
Then $q$ can be written a univariate Laurent polynomial, with maximum 
exponent at most $2T+R$\ and minimum exponent at least $-R$.
\end{lemma}

\begin{proof}
Let $|\psi_\mathrm{initial}\>$ denote the initial state of the algorithm,
which we can write as
\begin{align*}
|\psi_\mathrm{initial}\> &=
\left\vert S\right\rangle ^{\otimes R_1}
=\left(\frac{1}{\sqrt{|S|}} \sum_{i\in S}|i\>\right)^{\otimes R_1}
=\left(\frac{1}{\sqrt{|S|}} \sum_{i\in [N]}x_i|i\>\right)^{\otimes R_1}\\
&=\frac{1}{\left\vert S\right\vert ^{R_1/2}}\sum_{i_{1},\ldots,i_{R_1}\in\left[
N\right]  }x_{i_{1}}\cdots x_{i_{R_1}}\left\vert i_{1},\ldots,i_{R_1}\right\rangle.
\end{align*}
Thus, each amplitude is a complex multilinear polynomial in $X=\left(
x_{1},\ldots,x_{N}\right)  $ of degree $R_1$, divided by $\left\vert
S\right\vert ^{R_1/2}$.

Throughout the algorithm, each amplitude will remain a complex multilinear
polynomial in $X$ divided by some power of $|S|$.
Since $x_{i}^{2}=x_{i}$\ for all $i$, we can always
maintain multilinearity without loss of generality.

Like Beals et al.\ \cite{bbcmw}, we now consider how the polynomial degree of each amplitude and the power of $|S|$ in the denominator change as the algorithm progresses.
We have to handle 3 different kinds of unitaries that the quantum circuit may use:
the membership query oracle $\mathcal{O}_S$, unitaries independent of the input,
and the reflection unitary $\R_S$.

The first two cases are handled as in Beals et al.\
Since $\mathcal{O}_S$ is a unitary whose entries are degree-1 polynomials
in $X$, each use of this unitary increases a particular
amplitude's degree as a polynomial by $1$ and does not change the power of $|S|$ in the denominator. \
Second, input-independent unitary transformations only take linear
combinations of existing polynomials and hence do not increase the degree of the amplitudes or the power of $|S|$ in the denominator.
Finally, we consider the reflection unitary $\R_S = \id - 2|S\>\<S|$.
The $(i,j)^\mathrm{th}$ entry of this operator is $\delta_{ij} - \frac{2x_ix_j}{|S|} = \frac{\delta_{ij}|S| - 2x_ix_j}{|S|}$, where $\delta_{ij}$ is the Kronecker delta function.
Since $|S| = \sum_i x_i$, this is a degree-2 polynomial divided by $|S|$.
Hence applying this unitary will increase the degree of the amplitudes by $2$ and increase
the power of $|S|$ in the denominator by $1$.

In conclusion, we start with each amplitude being a polynomial of degree $R_1$ divided by $|S|^{R_1/2}$. $T$ queries to the membership oracle will increase the degree of each amplitude by at most $T$ and leave the power of $|S|$ in the denominator unchanged. $R_2$ uses of the reflection unitary will increase the degree by at most $2R_2$ and the power of $|S|$ in the denominator by $R_2$. It follows that $Q$'s final state has the form%
\begin{equation}
\left\vert \psi_\mathrm{final}\right\rangle =\sum_z \alpha_{z}\left(  X\right)
\left\vert z\right\rangle ,
\end{equation}
where each $\alpha_{z}\left(  X\right)  $\ is a complex multilinear polynomial
in $X$ of degree at most $R_1+2R_2+T = R+T$, divided by $\left\vert S\right\vert
^{R_1/2+R_2} = |S|^{R/2}$. \ Since $X$ itself is real-valued, it follows that the real and
imaginary parts of $\alpha_{z}\left(  X\right)  $, considered individually,
are real multilinear polynomials in $X$\ of degree at most $R+T$\ divided by
$\left\vert S\right\vert ^{R/2}$.

Hence, if we let%
\begin{equation}
p\left(  X\right)  :=\Pr\left[  Q^{\mathcal{O}_{S},\R_S}\left(  \left\vert
S\right\rangle ^{\otimes R_1}\right)  \text{ accepts}\right]  ,
\end{equation}
then%
\begin{equation}
p\left(  X\right)  =\sum_{\text{accepting }z}\left\vert \alpha_{z}\left(
X\right)  \right\vert ^{2}=\sum_{\text{accepting }z}\left(  \operatorname{Re}%
^{2}\alpha_{z}\left(  X\right)  +\operatorname{Im}^{2}\alpha_{z}\left(
X\right)  \right)
\end{equation}
is a real multilinear polynomial in $X$ of degree at most $2\left(
R+T\right)  $, divided through (in every monomial) by $\left\vert S\right\vert
^{R}=\left\vert X\right\vert ^{R}$.

Now consider%
\begin{equation}
q\left(  k\right)  :=\E_{\left\vert X\right\vert =k}\left[
p\left(  X\right)  \right]  .
\end{equation}
By \Cref{symlem}, this is a real univariate polynomial in $\left\vert
X\right\vert $ of degree at most $2\left(  R+T\right)  $, divided through (in
every monomial) by $\left\vert S\right\vert ^{R}=\left\vert X\right\vert ^{R}%
$. \ Or said another way, it's a real Laurent polynomial in $\left\vert
X\right\vert $, with maximum exponent at most $R+2T$\ and minimum exponent at
least $-R$.
\end{proof}

\subsection{Upper bounds}
\label{UPPER}\label{sec:upper}

Before proving our lower bounds on the degree of 
Laurent polynomials approximating $\mathsf{ApxCount}_{N,w}$, we
establish some simpler \emph{upper bounds}.
We show upper bounds on Laurent polynomial degree and in the 
queries, samples, and reflections model.

\para{Laurent polynomial degree of approximate counting.}
We now describe a \emph{purely negative} degree Laurent polynomial of degree $O(w^{1/3})$ for
approximate counting. \ This upper bound will serve as an important source
of intuition when we prove the (matching) lower bound of \Cref{thm:main}
(see \Cref{s:nextsubsection}). We are thankful to user \textquotedblleft
fedja\textquotedblright\ on MathOverflow for describing this construction.%
\footnote{%
See %
\url{https://mathoverflow.net/questions/302113/real-polynomial-bounded-at-inverse-integer-points}%
}

\begin{lemma}[fedja]
\label{fedjatight} For all $w$, there is a real polynomial $p$ of degree $O\left(w^{1/3}\right)$ such that:
\begin{enumerate}
\item $0 \le p(1/k) \le \frac{1}{3}$ for all $k \in [w]$.
\item $\frac{2}{3} \le p(1/k) \le 1$ for all integers $k \ge 2w$.
\item $0 \le p(1/k) \le 1$ for all $k \in \{w+1,w+2,\ldots,2w-1\}$.
\end{enumerate}
\end{lemma}

\begin{proof}
Assuming for simplicity that $w$ is a perfect cube, consider%
\begin{equation}
u\left(  x\right)  :=\left(  1-x\right)  \left(  1-2x\right)  \cdots\left(
1-w^{1/3}x\right)  .
\end{equation}
Notice that $\deg\left(  u\right)  =w^{1/3}$\ and $u\left(  \frac{1}%
{k}\right)  =0$ for all $k\in\left[  w^{1/3}\right]  $. \ Furthermore, we have
$u\left(  x\right) \in [0,1]$\ for all $x\in\left[
0,\frac{1}{w^{1/3}}\right]  $, and also $u\left(  x\right)  \in\left[
1-O\left(  \frac{1}{w^{1/3}}\right)  ,1\right]  $\ for all $x\in\left[
0,\frac{1}{w}\right]  $. \ Now, let $v$\ be the Chebyshev polynomial of degree
$w^{1/3}$, affinely adjusted so that $v\left(  x\right)
\in [0,1]$\ for all $x\in\left[  0,\frac{1}{w^{1/3}}\right]  $ (rather
than in $[-1,1]$ for all all $\left\vert x\right\vert \leq1$), and with a large jump between
$\frac{1}{2w}$\ and $\frac{1}{w}$. \ Then the product, $p(x):=u\left(  x\right)  v\left(  x\right)  $, has degree $2w^{1/3}$\ and
satisfies all the requirements, except possibly that the constants $\frac{1}{3}$ and $\frac{2}{3}$ in the first two requirements may be off. Composing with a constant degree polynomial corrects this, and gives a polynomial of degree $O(w^{1/3})$ that satisfies all three requirements.
\end{proof}

Interestingly, if we restrict our attention to purely negative degree
Laurent polynomials, then a matching lower bound is not too hard to show. \
In the same MathOverflow post, user fedja also proves the following, which can also
be shown using earlier work of Zhandry \cite[Proof of Theorem 7.3]{zhandry}):

\begin{lemma}
\label{fedjalem}Let $p$\ be a real polynomial, and suppose that $\left\vert
p\left( 1/k\right) \right\vert \leq1$ for all $k\in\left[ 2w\right] $, and
that $p\left( \frac{1}{w}\right) \leq\frac{1}{3}$\ while $p\left( \frac{1}{2w%
}\right) \geq\frac{2}{3}$. \ Then $\deg\left( p\right) =\Omega\left(
w^{1/3}\right) $.
\end{lemma}

\Cref{s:explosion} and \Cref{s:dualpoly} below take the considerable step of
extending \Cref{fedjalem} from purely negative degree Laurent
polynomials to general Laurent polynomials.

\para{Upper bounds in the queries, samples, and reflections model.} 
Although we showed that there is a purely negative degree Laurent polynomial of degree
$O(w^{1/3})$ for $\ApxCount$, this does not imply the existence of a quantum algorithm
in the queries, samples, and reflections model with similar complexity.

We now show that our lower bounds in the queries, samples, and reflections model (in \Cref{thm:main}) are tight (up to constants). This is \Cref{thm:alg} in the introduction, restated here for convenience:

\alg*

\begin{proof}
We describe two quantum algorithms for this problem with the two stated complexities.

The first algorithm uses $O(w^{1/3})$ samples and reflections. This algorithm is reminiscent of the original collision finding algorithm of Brassard, H{\o}yer, and Tapp~\cite{BHT98}. We first use $O(w^{1/3})$ copies of $|S\>$ to learn a set $M\subset S$ of size $w^{1/3}$ by simply measuring copies of $|S\>$ in the computational basis. Now we know that the ratio $|S|/|M|$ is either $w^{2/3}$ or $2w^{2/3}$. Now consider running Grover's algorithm on the set $S$ where the elements in $M$ are considered the ``marked'' elements. Grover's algorithm alternates reflections about the uniform superposition over the set being searched, $S$, with an operator that reflects about the marked elements in $M$. The first reflection is simply $\R_S$, which we have access to. The second unitary can be constructed since we have an explicit description of the set $M$. Now Grover's algorithm can be used to distinguish whether the fraction of marked elements is $1/w^{2/3}$ or half of that, and the cost will be $O(w^{1/3})$.

The second algorithm uses $O(\sqrt{{N}/{w}})$ reflections only and no copies of $|S\>$. 
Consider running the standard approximate counting algorithm~\cite{BHMT02} that uses membership queries to $S$ and distinguishes $|S|\leq w$ from $|S|\geq 2w$ using $O(\sqrt{N/w})$ membership queries. 
Observe that this algorithm starts with the state $|\psi\> =\frac{1}{%
\sqrt{N}}\left( \left\vert 1\right\rangle +\cdots +\left\vert N\right\rangle
\right)$, which is in $\textrm{span}\{|S\>,|\bar{S}\>\}$, and only uses reflections about $|\psi\> $ and membership queries to $|S\>$ in the form of a unitary that maps $|i\>$ to $-|i\>$ when $i \in S$.
This means the state of the algorithm remains in $\textrm{span}\{|S\>,|\bar{S}\>\}$
at all times. 
Within this subspace, a membership query to $S$ is the same as a reflection about $|S\>$.
Hence we can replace membership queries with the reflection operator to get an approximate counting algorithm that only uses $O(\sqrt{N/w})$ reflections and no copies of $|S\>$.
\end{proof}

Note that both the algorithms presented above generalize to the situation
where we want to distinguish $|S|=w$ from $|S|=(1+\varepsilon )w$. \ For the
first algorithm, we now pick a subset $M$ of size $w^{1/3}/\varepsilon ^{2/3}
$. Now we want to $(1+\varepsilon )$-approximate the fraction of marked
elements, which is either $1/(w\varepsilon )^{2/3}$ or $(1+\varepsilon )^{-1}
$ times that. This can be done with approximate counting~\cite[Theorem 15]%
{BHMT02}, and the cost will be $O\left( \frac{1}{\varepsilon }(w\varepsilon
)^{1/3}\right) =O\left( \frac{w^{1/3}}{\varepsilon ^{2/3}}\right) $.
The second algorithm is simpler to generalize, since we simply plug in
the query complexity of $\varepsilon$-approximate counting, which is 
$O\Bigl( \frac{1}{\varepsilon }\sqrt{\frac{N}{w}}\Bigr) $.

\subsection{Lower bound using the explosion argument}
\label{s:explosion}\label{sec:explosion}

We now show a weaker version of \Cref{thm:main} using the explosion argument
described in the introduction. \ The difference between the following
theorem and \Cref{thm:main} is the exponent of $w$ in the lower bound.

\begin{theorem}
\label{thm:explosion}Let $Q$ be a quantum algorithm that makes $T$ queries
to the membership oracle for $S$, and uses a total of $R$ copies of $|S\>$
and reflections about $|S\>$. \ If $Q$ decides whether $\left\vert
S\right\vert =w$\ or $\left\vert S\right\vert =2w$ with success probability
at least $2/3$, promised that one of those is the case, then either
\begin{equation}
T=\Omega \left( \sqrt{\frac{N}{w}}\right) \qquad \text{or}\qquad R=\Omega
\left( \min \left\{ w^{1/4},\sqrt{\frac{N}{w}}\right\} \right) .
\end{equation}
\end{theorem}

\begin{proof}
Since we neglect multiplicative constants in our lower bounds, let us allow the algorithm to use up to $R$ copies of $|S\>$ and $R$ uses of $\R_S$.
Let%
\begin{equation}
q\left(  k\right)  :=\E_{\left\vert S\right\vert =k}\left[
\Pr\left[  Q^{\mathcal{O}_{S},\R_S}\left(  \left\vert S\right\rangle ^{\otimes
R}\right)  \text{ accepts}\right]  \right]  .
\end{equation}
Then by \Cref{laurentlem}, we can write $q$ as a Laurent polynomial:
\begin{equation}
q\left(  k\right)  =u\left(  k\right)  +v\left(  1/k\right)  ,
\end{equation}
where $u$\ is a real polynomial in $k$ with $\deg\left(  u\right)  = O(T+R)$,
and $v$\ is a real polynomial in $1/k$\ with $\deg\left(  v\right) = O(R)$.
\ So to prove the theorem, it suffices to show that either $\deg\left(
u\right)  =\Omega\left(  \sqrt{\frac{N}{w}}\right)  $, or else $\deg\left(
v\right)  =\Omega\left(  w^{1/4}\right)  $. \ To do so, we'll assume that
$\deg\left(  u\right)  =o\left(  \sqrt{\frac{N}{w}}\right)  $\ and
$\deg\left(  v\right)  =o\left(  w^{1/4}\right)  $, and derive a contradiction.

Our high-level strategy is as follows: we'll observe that, if approximate
counting is being successfully solved, then either $u$ or $v$ must attain a
large first derivative somewhere in its domain. \ By the approximation theory
lemmas that we proved in \Cref{s:approxprelim}, this will force that polynomial to have a large
range---even on a subset of integer (or inverse-integer) points. \ But the
sum, $u\left(  k\right)  +v\left(  1/k\right)  $, is bounded in $\left[
0,1\right]  $\ for all $k\in\left[  N\right]  $. \ So if one polynomial has a
large range, then the other does too. \ But this forces the \textit{other}
polynomial to have a large derivative somewhere in its domain, and therefore
(by approximation theory) to have an even larger range, forcing the first
polynomial to have an even larger range to compensate, and so on. \ As long as
$\deg\left(  u\right)  $\ and $\deg\left(  v\right)  $ are both small enough,
this endless switching will force both $u$ and $v$ to attain
\textit{unboundedly }large values---with the fact that one polynomial is in
$k$, and the other is in $1/k$, crucial to achieving the desired
\textquotedblleft explosion.\textquotedblright\ \ Since $u$ and $v$ are
polynomials on compact sets, such unbounded growth is an obvious absurdity,
and this will give us the desired contradiction.

In more detail, we will study the following quantities.%
\begin{equation}%
\begin{tabular}
[c]{ll}%
$G_{u}:=\max_{x,y\in\left[  \sqrt{w},2w\right]  }\left\vert u\left(  x\right)
-u\left(  y\right)  \right\vert ~\ \ \ \ \ \ $ & $G_{v}:=\max_{x,y\in\left[
\frac{1}{N},\frac{1}{w}\right]  }\left\vert v\left(  x\right)  -v\left(
y\right)  \right\vert $\\
$\Delta_{u}:=\max_{x\in\left[  \sqrt{w},2w\right]  }\left\vert u^{\prime
}\left(  x\right)  \right\vert $ & $\Delta_{v}:=\max_{x\in\left[  \frac{1}%
{N},\frac{1}{w}\right]  }\left\vert v^{\prime}\left(  x\right)  \right\vert
$\\
$H_{u}:=\max_{x,y\in\left[  \sqrt{w},N\right]  }\left\vert u\left(  x\right)
-u\left(  y\right)  \right\vert $ & $H_{v}:=\max_{x,y\in\left[  \frac{1}%
{N},\frac{1}{\sqrt{w}}\right]  }\left\vert v\left(  x\right)  -v\left(
y\right)  \right\vert $\\
$I_{u}:=\max_{x,y\in\left[  w,N\right]  }\left\vert u\left(  x\right)
-u\left(  y\right)  \right\vert $ & $I_{v}:=\max_{x,y\in\left[  \frac{1}%
{2w},\frac{1}{\sqrt{w}}\right]  }\left\vert v\left(  x\right)  -v\left(
y\right)  \right\vert $\\
$L_{u}:=\max_{x,y\in\left\{  w,\ldots,N\right\}  }\left\vert u\left(
x\right)  -u\left(  y\right)  \right\vert $ & $L_{v}:=\max_{x,y\in\left\{
\sqrt{w},\ldots,2w\right\}  }\left\vert v\left(  \frac{1}{x}\right)  -v\left(
\frac{1}{y}\right)  \right\vert $%
\end{tabular}
\end{equation}

We have $0\leq q\left(  k\right)  \leq1$\ for all $k\in\left[  N\right]  $,
since in those cases $q\left(  k\right)  $\ represents a probability. \ Since
$Q$ solves approximate counting, we also have $q\left(  w\right)  \leq\frac
{1}{3}$\ and $q\left(  2w\right)  \geq\frac{2}{3}$. \ This means in particular
that either

\begin{enumerate}
\item[(i)] $u\left(  2w\right)  -u\left(  w\right)  \geq\frac{1}{6}$, and
hence $G_{u}\geq\frac{1}{6}$, or else

\item[(ii)] $v\left(  \frac{1}{2w}\right)  -v\left(  \frac{1}{w}\right)
\geq\frac{1}{6}$, and hence $G_{v}\geq\frac{1}{6}$.
\end{enumerate}

We will show that either case leads to a contradiction.

We have the following inequalities regarding $u$:%
\begin{equation}%
\begin{tabular}
[c]{ll}%
$G_{u}\geq L_{v}-1$ & by the boundedness of $q$\\
$\Delta_{u}\geq\frac{G_{u}}{2w}$ & by basic calculus\\
$H_{u}\geq\frac{\Delta_{u}\left(  N-\sqrt{w}\right)  }{\deg\left(  u\right)
^{2}}$ & by \Cref{markovlem}\\
$I_{u}\geq\frac{H_{u}}{2}$ & by \Cref{paturicor}\\
$L_{u}\geq\frac{I_{u}}{2}$ & by \Cref{ezrclem}%
\end{tabular}
\end{equation}
Here the fourth inequality uses the fact that, setting $\varepsilon
:=\frac{\sqrt{w}}{N}$, we have $\deg\left(  u\right)  =o\left(  \frac{1}%
{\sqrt{\varepsilon}}\right)  $ (thereby satisfying the hypothesis of \Cref{paturicor}), while the fifth inequality uses the fact that $\deg\left(
u\right)  =o\left(  \sqrt{N}\right)  $.

Meanwhile, we have the following inequalities regarding $v$:%
\begin{equation}%
\begin{tabular}
[c]{ll}%
$G_{v}\geq L_{u}-1$ & by the boundedness of $q$\\
$\Delta_{v}\geq G_{v}w$ & by basic calculus\\
$H_{v}\geq\frac{\Delta_{v}\left(  \frac{1}{\sqrt{w}}-\frac{1}{N}\right)
}{\deg\left(  v\right)  ^{2}}$ & by \Cref{markovlem}\\
$I_{v}\geq\frac{H_{v}}{2}$ & by \Cref{paturicor}\\
$L_{v}\geq\frac{I_{v}}{2}$ & by \Cref{ezrclem}%
\end{tabular}
\end{equation}
Here the fourth inequality uses the fact that, setting $\varepsilon
:=\frac{1/2w}{1/\sqrt{w}}=\frac{1}{2\sqrt{w}}$, we have $\deg\left(  v\right)
=o\left(  \frac{1}{\sqrt{\varepsilon}}\right)  $ (thereby satisfying the
hypothesis of \Cref{paturicor}). \ The fifth inequality uses the fact
that, if we set $V\left(  x\right)  :=v\left(  x/w\right)  $, then the
situation\ satisfies the hypothesis of \Cref{ezrclem}:\ we are interested
in the range of $V$ on the interval $\left[  \frac{1}{2},\sqrt{w}\right]  $,
compared to its range on discrete points $\frac{w}{\sqrt{w}},\frac{w}{\sqrt
{w}+1},\ldots,\frac{w}{2w}$\ that are spaced at most $1$ apart from each
other; and we also have $\deg\left(  V\right)  =\deg\left(  v\right)
=o\left(  w^{1/4}\right)  $.

All that remains is to show that, if we insert either $G_{u}\geq\frac{1}{6}%
$\ or $G_{v}\geq\frac{1}{6}$ into the coupled system of inequalities above,
then we get unbounded growth and the inequalities have no solution. \ Let us
collapse the two sets of inequalities to%
\begin{align*}
L_{u}  &  \geq\frac{1}{4}\frac{N-\sqrt{w}}{\deg\left(  u\right)  ^{2}}%
\frac{G_{u}}{2w}=\Omega\left(  \frac{N}{w\deg\left(  u\right)  ^{2}}%
G_{u}\right)  ,\\
L_{v}  &  \geq\frac{1}{4}\frac{\frac{1}{\sqrt{w}}-\frac{1}{N}}{\deg\left(
v\right)  ^{2}}G_{v}w=\Omega\left(  \frac{\sqrt{w}}{\deg\left(  v\right)
^{2}}G_{v}\right)  .
\end{align*}
Hence%
\begin{align*}
G_{u}  &  \geq L_{v}-1=\Omega\left(  \frac{\sqrt{w}}{\deg\left(  v\right)
^{2}}G_{v}\right)  -1,\\
G_{v}  &  \geq L_{u}-1=\Omega\left(  \frac{N}{w\deg\left(  u\right)  ^{2}%
}G_{u}\right)  -1.
\end{align*}
By the assumption that $\deg\left(  v\right)  =o\left(  w^{1/4}\right)  $\ and
$\deg\left(  u\right)  =o\left(  \sqrt{\frac{N}{w}}\right)  $, we have
$\frac{\sqrt{w}}{\deg\left(  v\right)  ^{2}}\gg1$\ and $\frac{N}{w\deg\left(
u\right)  ^{2}}\gg1$. \ Plugging in $G_{u}\geq\frac{1}{6}$\ or $G_{v}\geq
\frac{1}{6}$, this is enough to give us unbounded growth.
\end{proof}

\subsection{Lower bound using dual polynomials}
\label{s:dualpoly}\label{sec:dualpoly}

In this section we use the method of dual polynomials to
establish our main result, \Cref{thm:main}, restated for convenience:

\main*

Let $p(r)$ be a univariate Laurent polynomial of negative degree $D_{1}$ and
positive degree $D_{2}$. \ That is, let $p(r)$ be of the form
\begin{equation}
p(r)=a_{0}/r^{D_{1}}+a_{1}/r^{D_{1}-1}+\dots
+a_{D_{1}-1}/r+a_{D_{1}}+a_{D_{1}+1}\cdot r+\cdots +a_{D_{2}+D_{1}}\cdot
r^{D_{2}}.
\end{equation}

\Cref{thm:main} follows by combining the Laurent polynomial method (\Cref{laurentlem}) and the following
theorem.

\begin{theorem}
\label{mainthm} 
Let $\varepsilon < 1$. Suppose that $p$ has negative degree $D_1$ and positive
degree $D_2$ and satisfies the following properties.

\begin{itemize}
\item $\left\vert p(w)-1\right\vert \leq \varepsilon $

\item $\left\vert p(2w)+1\right\vert \leq \varepsilon $

\item $\left\vert p(\ell )\right\vert \leq 1+\varepsilon \text{ for all }%
\ell \in \{1,2,\dots ,n\}$
\end{itemize}

Then either $D_1 \geq \Omega\left(w^{1/3}\right)$ or $D_2 \geq \Omega\left(\sqrt{N/w}%
\right)$.
\end{theorem}

In fact, our proof of \Cref{mainthm} will show that the lower bound holds
even if $\left\vert p(\ell )\right\vert \leq 1+\varepsilon $ only for $\ell
\in \{w^{1/3},w^{1/3}+1,\dots ,w\}\cup \{2w,2w+1,\dots ,N\}$. \ We refer to
a Laurent polynomial $p$ satisfying the three properties of \Cref{mainthm}
as an \emph{approximation for approximate counting}.

\paragraph{Proof of \Cref{mainthm}.}

Let $p$ be any Laurent polynomial satisfying the hypothesis of \Cref{mainthm}%
. \ We begin by transforming $p$ into a (standard) polynomial $q$ in a
straightforward manner. \ This transformation is captured in the following
lemma, whose proof is so simple that we omit it.

\begin{lemma}
\label{mainlem} If $p$ satisfies the properties of \Cref{mainthm}, then the
polynomial $q(r) = p(r) \cdot r^{D_1}= a_0 + a_1 r + \dots + a_{D_1 + D_2}
r^{D_1 + D_2}$ is a (standard) polynomial of degree at most $D_1 + D_2$, and
$q$ satisfies the following three properties.

\begin{itemize}
\item $\left\vert q(w)-w^{D_{1}}\right\vert \leq \varepsilon \cdot w^{D_{1}}$

\item $\left\vert q(2w)+(2w)^{D_{1}}\right\vert \leq \varepsilon \cdot
(2w)^{D_{1}}$

\item $\left\vert q(\ell )\right\vert \leq \left( 1+\varepsilon \right) \ell
^{D_{1}}\text{ for all }\ell \in \{1,2,\dots ,N\}$
\end{itemize}
\end{lemma}

We now turn to showing that, for any constant $\varepsilon < 1$, no
polynomial $q$ can satisfy the conditions of \Cref{mainlem} unless $D_1 \geq
\Omega(w^{1/3})$ or $D_2 \geq \Omega\left(\sqrt{N/w}\right)$.

Consider the following linear program. \ The variables of the linear program
are $\varepsilon $, and the $D_{2}+D_{1}+1$ coefficients of $q$.

\begin{equation}
\boxed{\begin{array}{ll} \text{minimize } & \eps \\ \mbox{such that} & \\
&|q(w) - w^{D_1}| \leq \eps \cdot w^{D_1}\\ &|q(2w) + (2w)^{D_1}| \leq \eps
\cdot (2w)^{D_1}\\ & |q(\ell)| \leq (1+\eps) \cdot \ell^{D_1} \text{ for all
} \ell \in \{1, 2, \dots, N\}\\ &\eps \geq 0 \end{array}}
\end{equation}

Standard manipulations reveal the dual.

\begin{equation} \label{duallp}
\boxed{\begin{array}{ll} \text{maximize } \phi(w) \cdot w^{D_1} - \phi(2w)
\cdot (2w)^{D_1} - \sum_{\ell \in \{1, \dots, N\}, \ell \not\in \{w, 2w\}}
|\phi(\ell)| \cdot \ell^{D_1} \\ \mbox{such that} & \\ \sum_{\ell=1}^{N}
\phi(\ell) \cdot \ell^j =0 \text{ for } j= 0, 1, 2, \dots, D_1 + D_2 &\\
\sum_{\ell=1}^N |\phi(\ell)| \cdot \ell^{D_1} = 1\\ \phi \colon \Reals \to
\mathbb{R} \end{array}}
\end{equation}

\label{s:duallp}

\medskip\medskip

\Cref{mainthm} will follow if we can exhibit a solution $\phi $ to the dual
linear program achieving value $\varepsilon >0$, for some
setting of $D_{1}\geq \Omega (w^{1/3})$ and $D_{2}\geq \Omega \left( \sqrt{%
N/w}\right) $.\footnote{%
We will alternatively refer to such dual solutions $\phi $ as \emph{dual
witnesses}, since they act as a witness to the fact that any low-degree
Laurent polynomial $p$ approximating the approximate counting problem must
have large error.} \ We now turn to this task.

\subsubsection{Constructing the dual solution}
\label{duallpsec}

For a set $T\subseteq \{0,1,\dots ,N\}$, define
\begin{equation}
Q_{T}(t)=\prod_{i=0,1,\dots ,N,i\not\in T}(t-i).
\end{equation}%
Let $c>2$ be an integer constant that we will choose later (the bigger we choose $c$
to be, the better the objective value achieved by our final dual witness. \
But choosing a bigger $c$ will also lower the degrees $D_{1},D_{2}$ of
Laurent polynomials against which our lower bound will hold).

We now define two sets $T_{1}$ and $T_{2}$. \ The size of $T_{1}$ will be
\begin{equation}
d_{1}:=\lfloor \left( w/c\right) ^{1/3}\rfloor =\Theta \left( w^{1/3}\right)
\end{equation}%
and the size of $T_{2}$ will be $d_{2}$ for
\begin{equation}
d_{2}:=\lfloor \sqrt{N/(cw)}\rfloor =\Theta \left( \sqrt{N/w}\right) .
\end{equation}%

Let
\begin{equation}
T_1 = \left\{\lfloor w/(c i^2) \rfloor \colon i=1, 2, \dots, d_1\right\}
\end{equation}
and
\begin{equation}
T_2 = \left\{c \cdot i^2 \cdot w \colon i = 1, 2, \dots, d_2:=\sqrt{N/(cw)}%
\right\}.
\end{equation}
Finally, define
\begin{equation}
T = \{w, 2w\} \cup T_1 \cup T_2.
\end{equation}

At last, define $\Phi \colon \{0, 1, \dots, N\} \to \mathbb{R}$ via
\begin{equation}  \label{phidef}
\Phi(t) = (-1)^t \cdot \binom{N}{t} \cdot Q_T(t).
\end{equation}

Our final dual solution $\phi $ will be a scaled version of $\Phi $. \
Specifically, $\Phi $ itself does not satisfy the second constraint of the
dual linear program, that $\sum_{\ell =1}^{N}|\Phi (\ell )|\cdot \ell
^{D_{1}}=1$. \ So letting
\begin{equation}
C=\sum_{\ell =1}^{N}|\Phi (\ell )|\cdot \ell ^{D_{1}},  \label{Cdef}
\end{equation}%
our final dual witness $\phi $ will be $\Phi /C$.

\para{The sizes of $T_{1}$ and $T_{2}$.} \ Clearly,
under the above definition of $T_{2}$, $|T_{2}|=d_{2}$ as claimed above. \
It is not as immediately evident that $|T_{1}|=d_{1}$: to establish this, we
must show that for distinct $i,j\in \{1,2,\dots ,d_{1}\}$, $\lfloor
w/(ci^{2})\rfloor \neq \lfloor w/(cj^{2})\rfloor $. \ This is handled in the
following easy lemma.

\begin{lemma}
\label{lem:distinct} Let $i\neq j$ be distinct numbers in $\{1,\dots ,d_{1}\}
$ and $c>2$ be a constant. \ Then as long as $d_{1}<\left( w/c\right) ^{1/3}$%
, it holds that $\lfloor w/(ci^{2})\rfloor \neq \lfloor w/(cj^{2})\rfloor $.
\end{lemma}

\begin{proof}
Assume without loss of generality that  $i> j$. Then $w/(c j^2)- w/(c i^2)$ is clearly minimized when $i=d_1$ and $j=i-1$.
For the remainder of the proof, fix $i=d_1$.
In this case,
\begin{align}
w/(c j^2) - w/(c i^2) \geq
w/\left(c (i-1)^2\right) - w/\left(c i^2\right)
= \frac{w i^2 - w \left(i-1\right)^2}{c \cdot i^2 \cdot \left(i-1\right)^2} \notag \\
= \frac{w}{c} \cdot \frac{2i-1}{i^2\left(i-1\right)^2 }
\geq \frac{w}{c} \cdot \frac{2i-1}{i^4}
\geq \frac{w}{c i^3} \geq 1. \label{eq28}
\end{align}
Here, the final inequality holds because $i^3 = d_1^3 \leq w/c$.

\Cref{eq28} implies the lemma, as two numbers whose difference is at least 1 cannot have
the same integer floor.
\end{proof}

\Cref{lem:distinct} is false for $d_1 = \omega(w^{1/3})$, highlighting on a
technical level why one cannot choose $d_1$ larger than $\Theta(w^{1/3})$
without the entire construction and analysis of $\Phi$ breaking down.

\subsubsection{Intuition: ``gluing together'' two simpler dual solutions}

Before analyzing the dual witnesses
$\Phi$ and $\phi$ constructed in \Cref{phidef} and \Cref{Cdef},
in this subsection and the next, we provide detailed intuition
for why the definitions of $\Phi$ and $\phi$ are natural,
and briefly overview their analysis. 

\label{sec:intuition}

\para{A dual witness for purely positive degree (i.e.,
approximate degree).} \ Suppose we were merely interested in showing an
approximate degree lower bound of $\Omega (\sqrt{N/w})$ for approximate
counting (i.e., a lower bound on the degree of traditional polynomials that
distinguish input $w$ from $2w$, and are bounded at all other integer inputs
in $1,\dots ,N$). \ This is equivalent to exhibiting a solution to the dual
linear program with $D_{1}=0$. \ A valid dual witness $\phi_1$ for this
simpler case is to also use \Cref{phidef}, but to set

\begin{equation}  \label{simplerT1}
T= \{w, 2w\} \cup T_2,
\end{equation}
rather than $T=\{w, 2w\} \cup T_1 \cup T_2$.

We will explain intuition for why \Cref{simplerT1} is a valid dual
solution for the approximate degree of approximate counting in the next
subsection. \ For now, we wish to explain how this construction relates to
prior work. In \cite{bt13}, for any constant $\delta >0$, a dual
witness is given for the fact that the $(1-\delta )$-approximate degree
of $\mathsf{OR}$ is $\Omega (\sqrt{N})$. \ This dual witness \emph{nearly}
corresponds to the above, with $w=1$. \ Specifically, Bun and Thaler \cite{bt13} use the
set $T=\{0,1\}\cup \{ci^{2}\colon i=1,2,\dots ,\sqrt{N/c}\}$, and they show
that almost all of the \textquotedblleft mass\textquotedblright\ of this
dual witness is located on the inputs $0$ and $1$, i.e.,
\begin{equation}
\left\vert \Phi (0)\right\vert +\left\vert \Phi (1)\right\vert \geq
(1-\delta )\cdot \sum_{i=2}^{N}\left\vert \Phi (i)\right\vert .
\label{epseq}
\end{equation}%
Here, the bigger $c$ is chosen to be, the smaller the value of $\delta$
for which \Cref{epseq} holds.

In the case of $w=1$, our dual witness for approximate counting differs from
this only in that $\{0, 1\}$ is replaced with $\{1, 2\}$. This is
because, in order to show a lower bound for distinguishing input $w=1$ from
input $2w=2$, we want almost all of the mass to be on inputs $\{1, 2\}$
rather than $\{0, 1\}$ (this is what will ensure that the objective function
of the dual linear program is large).

For general $w$, we want most of the mass of $\psi $ to be concentrated on
inputs $w$ and $2w$. \ Accordingly, relative to the $w=1$ case, we
effectively multiply \emph{all} points in $T$ by $w$, and one can show that
this does not affect the calculation regarding concentration of mass.

\para{A dual witness for purely negative degree.} \
Now, suppose we were merely interested in showing that Laurent polynomials
of \emph{purely negative} degree require degree $\Omega (w^{1/3})$ to
approximate the approximate counting problem. \ This is equivalent to
exhibiting a solution to the dual linear program with $D_{2}=0$.  Then a
valid dual witness $\phi_2 $ for this simpler case is to also use 
\Cref{phidef}, but to set

\begin{equation}  \label{simplerT2}
T= \{w, 2w\} \cup T_1.
\end{equation}

Again, we will give intuition for why this is a valid dual solution in the
next subsection (\Cref{s:nextsubsection}). \ For now, we wish to explain how
this construction relates to prior work. \ Essentially, the $\Omega \left(
w^{1/3}\right) $-degree lower bound for Laurent polynomials with \emph{only
negative} powers was proved by Zhandry \cite[Theorem 7.3]{zhandry}. \
Translating Zhandry's theorem into our setting is not entirely trivial, and
he did not explicitly construct a solution to our dual linear program. \
However (albeit with significant effort), one can translate his argument to
our setting to show that \Cref{simplerT2} gives a valid dual
solution to prove a lower bound against Laurent polynomials with only
negative powers.

\para{Gluing them together.} The above discussion
explains that the key ideas for constructing dual solutions $\phi _{1}$, $%
\phi _{2}$ witnessing degree lower bounds for Laurent polynomials of \emph{%
only negative} or \emph{only positive} powers were essentially already
known, or at least can be extracted from prior work with enough effort. \ In this work, we are interested in proving lower bounds for Laurent
polynomials with both positive and negative powers. \ Our dual solution $%
\Phi $ essentially just \textquotedblleft glues together\textquotedblright\
the dual solutions that can be derived from prior work. \ By this, we mean
that the set $T$ of integer points on which our $\Phi $ is nonzero is the
\emph{union} of the corresponding sets for $\phi _{1}$ and $\phi _{2}$
individually. \ Moreover, this union is nearly disjoint, as the only points
in the intersection of the two sets being unioned are $w$ and $2w$.

\para{Overview of the analysis.}
To show that we have constructed
a valid solution 
to the dual linear program (\Cref{duallp}),
we
must establish
that (a) $\Phi$ is uncorrelated
with every polynomial of degree 
at most $D_1 + D_2$
and (b) $\Phi$ is
well-correlated  with 
any function $g$ that evaluates
to $+1$ on input $w$, to $-1$ 
on input $2w$, and is bounded in $[-1, 1]$ elsewhere.
In (b), the correlation is taken
with respect to 
an appropriate weighting
of the inputs, that on input $\ell \in [N]$ places
mass proportional to $\ell^{D_1}$.

The definition of $\Phi$
as a ``gluing together'' of $\phi_1$ and $\phi_2$
turns out, in a
straightforward manner, to ensure
that $\Phi$ is uncorrelated
with polynomials of degree at 
$D_1 + D_2$. 
All that remains is to
show that $\Phi$ is well-correlated
with $g$ under the appropriate
weighting of inputs. 
This turns out to be technically demanding, but ultimately
can be understood as stemming from the fact that 
$\phi_1$ and $\phi_2$ are individually
well-correlated with $g$ (albeit, in the case of $\phi_2$, under a
\emph{different} weighting of the inputs than the weighting that is relevant for $\Phi$). 

\subsubsection{Intuition via complementary slackness}

\label{s:nextsubsection}We now attempt to lend some insight into why the
dual witnesses $\phi_1 $ and $\phi_2 $ for the purely positive degree and purely
negative degree take the form that they do. \ This section is deliberately
slightly imprecise in places, and builds on intuition that has been put
forth in prior works proving approximate degree lower bounds via dual
witnesses \cite{bt13, Tha16, bkt}.

Notice that $\phi_1 $ is precisely defined so that $\phi_1 (i)=0$ for any $%
i\not\in \{w,2w\}\cup T_{2}$, and similarly $\phi_2 (i)=0$ for any $i\not\in
\{w,2w\}\cup T_{1}$. \ The intuition for why this is reasonable comes from
complementary slackness, which states that an optimal dual witness should
equal $0$ except on inputs that correspond to primal constraints that are
\emph{made tight by an optimal primal solution}. \ By \textquotedblleft
constraints made tight by an optimal primal solution\textquotedblright , we
mean constraints that, for the optimal primal solution, hold with equality
rather than (strict) inequality.

Unpacking that statement, this means the following. \ Suppose that $q$ is an
optimal solution to the primal linear program of \Cref{s:duallp}, meaning it
minimizes the error $\varepsilon $ amongst all polynomials of the same same
degree. \ The constraints made tight by $q$ are precisely those inputs $\ell
$ at which $q$ hits its \textquotedblleft maximum error\textquotedblright\
(e.g., an input $\ell$ such that $|q(\ell)|=(1+\varepsilon )\cdot \ell ^{D_{1}}$).
\ We call these inputs \emph{maximum-error} inputs for $q$. \ Complementary
slackness says that there is an optimal solution to the dual linear program (\Cref{duallp})
that equals $0$ at all inputs that are not maximum-error inputs for $q$.

In both the purely positive degree case, and the purely negative degree
case, we know roughly what primal optimal solutions $q$ look like, and
moreover we know what roughly their maximum-error points look like. \ In the
first case, the maximum-error points are well-approximated by the points in $%
T_{2}$, and in the purely negative degree case, the maximum error points are
well-approximated by the points in $T_{1}$. \ Let us explain.

\para{Purely positive degree case.} \ Let $T_{d}$ be
the degree $d$ Chebyshev polynomial of the first kind.  It can be seen that $%
P(\ell )=T_{\sqrt{N}}\left( 1+2/N-\ell /N\right) $ satisfies $P(1)\geq 2$,
while $|P(\ell )|\leq 1$ for $\ell =2,3,\dots ,N$. \ That is, up to scaling,
$P$ approximates the approximate counting problem for $w=1$, and its known
that its degree is within a constant factor of optimal.

It is known that the extreme points of $T_{d}$ are of the following form,
for $k=1,\dots ,d$:
\begin{equation}
\cos \left( \frac{(2k-1)}{2d}\pi \right) \approx 1-k^{2}/(2d^{2}),
\label{tester}
\end{equation}%
where the approximation uses the Taylor expansion of the cosine function
around $0$. \ \Cref{tester} means that the extreme points of $P$
are roughly those inputs $\ell $ such that $1+2/N-\ell /N\approx
1-k^{2}/(2d^{2})$, where $d=\sqrt{N}$. \ Such $\ell $ are roughly of the
form $\ell \approx c\cdot i^{2}$ for some constant $c$, as $i$ ranges from $1
$ up to $\Theta (N^{1/2})$.

More generally, when $w\geq 1$, an asymptotically optimal approximation for
distinguishing input $w$ from $2w$ is $P(\ell )=T_{\sqrt{N/w}}\left(
1+2w/N-\ell /(wN)\right) $. \ The extreme points of $P$ are roughly of the
form $\ell \approx c\cdot i^{2}\cdot w$ for some constant $c$, as $i$ ranges
from $1$ up to $\Theta (\sqrt{N/w})$, which is exactly the form of the
points in our set $T_{2}$.

\para{Purely negative degree case.} In \Cref{fedjatight}%
, we exhibited a simple, purely negative degree Laurent polynomial $p$
(i.e., $p(\ell )$ is a standard polynomial in $1/\ell $) with degree $%
D_{1}=w^{1/3}$ that solves the approximate counting problem (the
construction is due to MathOverflow user \textquotedblleft
fedja\textquotedblright ). \ Roughly speaking, $p$ can be written as a
product $p(\ell )=u(\ell )\cdot v(\ell )$, where $u(\ell )$ has the roots $%
\ell =1,2,\dots ,w^{1/3}$, and $v(\ell )$ is (an affine transformation) of a
Chebyshev polynomial of degree $w^{1/3}$, applied to $1/\ell $. \ One can
easily look at this construction and see that $p(\ell )$ outputs \emph{%
exactly} the correct value on inputs $\{1,2,\dots ,w^{1/3}\}$, so these are
not maximum error points for $p$. \ Moreover, the analysis of the maximum
error points for Chebyshev polynomials above can be applied to show that the
maximum error points of $p$ are roughly of the form $\ell $ such that $%
1/\ell =c\cdot i^{2}/w$ for some constant $c$, with $i$ ranging from $1$ up
to $\Theta (w^{1/3})$. \ This means that the extreme points are roughly of
the form $\ell \approx \frac{w}{ci^{2}}$, which is why our set $T_{1}$
consists of points of the form $\lfloor \frac{w}{ci^{2}}\rfloor $ (the
floors are required because we are proving lower bounds against polynomials
whose behavior is only constrained at integer inputs).

\subsubsection{Analysis of the dual solution \texorpdfstring{$\Phi$}{Phi}}

\begin{lemma}
\label{phdlem}Let $d_{1}=|T_{1}|$ and $d_{2}=|T_{2}|$. \ Then for any $%
j=0,1,\dots ,d_{1}+d_{2}$, it holds that
\begin{equation*}
\sum_{\ell =1}^{N}\Phi (\ell )\cdot \ell ^{j}=0.
\end{equation*}
\end{lemma}

\begin{proof}
A basic combinatorial fact is that for any polynomial $Q$ of degree at most $N-1$,
the following identity holds:
\begin{equation} 
\label{keyidentity} 
\sum_{\ell=0}^N \binom{N}{\ell} (-1)^\ell Q(\ell) = 0.
\end{equation}

Observe that for any $j \leq d_1 + d_2+1$,
\begin{equation} \label{isapolynomial}  
Q_T(\ell) \cdot \ell^j  \text{ is a polynomial in } \ell \text{ of degree at most } N-1. 
\end{equation}
Furthermore, $\Phi(0)=0$, because $0 \not\in T$. Hence
\begin{equation} \label{nozeroterm} 
\sum_{\ell=0}^N \binom{N}{\ell} (-1)^\ell Q_T(\ell) \cdot \ell^j= \sum_{\ell=1}^N \binom{N}{\ell} (-1)^\ell Q_T(\ell) \cdot \ell^j. 
\end{equation}
Thus, we can calculate:
\begin{eqnarray*}
\sum_{\ell=1}^N \Phi(\ell) \cdot \ell^{j} = \sum_{\ell=1}^N (-1)^{\ell} \cdot \binom{N}{\ell} \cdot Q_T(\ell) \cdot \ell^j\\
= \sum_{\ell=0}^N (-1)^\ell \cdot \binom{N}{\ell} \cdot Q_T(\ell) \cdot \ell^j
=0.
\end{eqnarray*}
Here, the second equality follows from \Cref{nozeroterm}, while the third follows
from Equations \eqref{keyidentity} and \eqref{isapolynomial}.
\end{proof}

Let us turn to analyzing $\Phi $'s value on various inputs. \ Clearly the
following condition holds:

\begin{equation}
\Phi (\ell )=0~\text{for all }\ell \not\in T\text{.}  \label{zeros}
\end{equation}%
Next, observe that for any $r\in T$,
\begin{equation*}
\left\vert \Phi (r)\right\vert =N!\cdot \frac{1}{\prod_{j\in T,j\neq r}|r-j|}%
.
\end{equation*}

Consider any quantity $c\cdot i^{2}\cdot w\in T_{2}$. \ Then
\begin{flalign} & \left|\Phi(c \cdot w \cdot i^2)\right|/\left|\Phi(w)\right| = \frac{\prod_{j \in T, j \neq w} |w- j|}{ \prod_{j \in T, j \neq c \cdot i^2 \cdot w} |w \cdot c \cdot i^2 - j|} \notag  \\
& = \frac{\left|w-2w\right| \cdot \left(\prod_{j=1}^{d_2} \left|w-c \cdot j^2 \cdot w\right|\right) \cdot \left(\prod_{j=1}^{d_1} \left(w-\left\lfloor \frac{w}{cj^2}\right\rfloor\right)\right)}{\left|c \cdot i^2 \cdot w - w\right| \cdot \left|c \cdot i^2 \cdot w - 2w\right| \cdot \left( \prod_{j=1, j\neq i}^{d_2} \left|w \cdot c \cdot i^2 - w \cdot c \cdot j^2\right| \right)\cdot \left( \prod_{j=1}^{d_1} \left(w \cdot c \cdot i^2 - \left\lfloor\frac{w}{c \cdot j^2}\right\rfloor\right)\right)}\notag \\
& = \frac{c^{d_2} \cdot \left(\prod_{j=1}^{d_2} \left( j^2 - \frac{1}{c}\right)\right) \cdot \prod_{j=1}^{d_1} \left(w-\left\lfloor\frac{w}{c\cdot j^2}\right\rfloor\right)}{ \left(ci^2 -1\right)\cdot \left(ci^2 -2 \right) \cdot c^{d_2-1} \cdot \left( \prod_{j=1, j \neq i}^{d_2} \left|i^2 - j^2\right|\right) \cdot \left(\prod_{j=1}^{d_1} \left(w \cdot c \cdot i^2 - \left\lfloor \frac{w}{c \cdot j^2}\right\rfloor\right)\right)} \notag \\
& \leq
\frac{ c \cdot \left( \prod_{j=1}^{d_2} \left( j^2 - \frac{1}{c}\right)\right) \cdot \prod_{j=1}^{d_1} \left(w-\left\lfloor\frac{w}{c\cdot j^2}\right\rfloor\right)}{ \left(ci^2 -1\right)\cdot \left(ci^2 -2 \right) \cdot \left(\prod_{j=1, j \neq i}^{d_2} \left|i^2 - j^2\right|\right) \cdot \left(\prod_{j=1}^{d_1} \left(w \cdot c \cdot i^2 -  \frac{w}{c \cdot j^2}\right)\right)} \label{firstlongequation}
\end{flalign}

Now, observe that
\begin{align}
\prod_{j=1}^{d_1} \left(w-\left\lfloor\frac{w}{c\cdot j^2}\right\rfloor\right)
\leq \prod_{j=1}^{d_1} \left(w - \frac{w}{cj^2} +1 \right) =
\prod_{j=1}^{d_1} w \cdot \left(1-\frac{1}{cj^2}\right) \cdot \left( 1+
\frac{1}{w \cdot \left(1-\frac{1}{cj^2}\right)}\right)  \notag \\
\leq \prod_{j=1}^{d_1} w \cdot \left(1-\frac{1}{cj^2}\right) \left(1 + \frac{%
1}{(1-1/c) \cdot w}\right) \leq \left( \prod_{j=1}^{d_1} w \cdot \left(1-%
\frac{1}{cj^2}\right)\right) \cdot \left(1+o(1) \right).  \label{alongequation}
\end{align}

Hence, we see that Expression \eqref{firstlongequation} is bounded by

\begin{align}
& \frac{c\cdot \left( \prod_{j=1}^{d_{2}}\left( j^{2}-\frac{1}{c}\right)
\right) \cdot \left( \prod_{j=1}^{d_{1}}\left( 1-\frac{1}{c\cdot j^{2}}%
\right) \right) \cdot (1+o(1))}{\left( ci^{2}-1\right) \cdot \left(
ci^{2}-2\right) \cdot \left( \prod_{j=1,j\neq i}^{d_{2}}\left\vert
i^{2}-j^{2}\right\vert \right) \cdot \left( \prod_{j=1}^{d_{1}}\left( c\cdot
i^{2}-\frac{1}{c\cdot j^{2}}\right) \right) }  \notag \\
& \leq \frac{c\cdot \left( d_{2}!\right) ^{2}\cdot \left(
\prod_{j=1}^{d_{1}}\left( 1-\frac{1}{c\cdot j^{2}}\right) \right) \cdot
(1+o(1))}{\left( ci^{2}-1\right) \cdot \left( ci^{2}-2\right) \cdot \left(
\prod_{j=1,j\neq i}^{d_{2}}\left\vert i-j\right\vert \left\vert
i+j\right\vert \right) \cdot (c\cdot i^{2})^{d_{1}}\cdot \left(
\prod_{j=1}^{d_{1}}\left( 1-\frac{1}{c^{2}\cdot i^{2}\cdot j^{2}}\right)
\right) }  \notag \\
& =\frac{c\cdot \left( d_{2}!\right) ^{2}\cdot 2i^{2}\cdot \left(
\prod_{j=1}^{d_{1}}\left( 1-\frac{1}{c\cdot j^{2}}\right) \right) \cdot
(1+o(1))}{\left( ci^{2}-1\right) \cdot \left( ci^{2}-2\right) \cdot \left(
d_{2}+i\right) !\left( d_{2}-i\right) !\cdot (c\cdot i^{2})^{d_{1}}\cdot
\left( \prod_{j=1}^{d_{1}}\left( 1-\frac{1}{c^{2}\cdot i^{2}\cdot j^{2}}%
\right) \right) }  \notag \\
& \leq \frac{c\cdot 2i^{2}\cdot \left( d_{2}!\right) ^{2}\cdot (1+o(1))}{%
\left( ci^{2}-1\right) \left( ci^{2}-2\right) \cdot \left( d_{2}+i\right)
!\left( d_{2}-i\right) !\cdot (c\cdot i^{2})^{d_{1}}}\leq \frac{2\left(
1+o(1)\right) }{\left( 1-\frac{1}{c\cdot i^{2}}\right) \cdot (c\cdot
i^{2}-2)\cdot (c\cdot i^{2})^{d_{1}}}.  \label{verylongequation}
\end{align}%
In the penultimate inequality, we used the fact that $\frac{(d_{2}!)^{2}}{%
(d_{2}+i)!(d_{2}-i)!}=\frac{\binom{2d_{2}}{d_{2}+i}}{\binom{2d_{2}}{d_{2}}}%
\leq 1$.

Next, consider any quantity $\left\lfloor \frac{w}{c\cdot i^{2}}\right\rfloor \in T_{1}$%
. \ Then
\begin{flalign} & \left|\Phi\left( \left\lfloor \frac{w}{c \cdot i^2}\right\rfloor\right)\right|/\left|\Phi(w)\right| \notag \\
& = \frac{|w-2w| \left(\prod_{j=1}^{d_2} |w - cj^2 w| \right) \left(\prod_{j=1}^{d_1} \left( w - \left\lfloor  \frac{w}{cj^2}\right\rfloor\right)\right)}{ \left(w-\left\lfloor \frac{w}{c \cdot i^2}\right\rfloor\right) \cdot \left(2w - \left\lfloor \frac{w}{c \cdot i^2}\right\rfloor \right) \left(\prod_{j=1}^{d_2} \left(w \cdot c \cdot j^2 - \left\lfloor \frac{w}{c \cdot i^2}\right\rfloor\right)\right) \prod_{j=1, j\neq i}^{d_1} \left|\left\lfloor \frac{w}{c \cdot i^2}\right\rfloor - \left\lfloor \frac{w}{c \cdot j^2}\right\rfloor\right|} \notag \\
& \leq
\frac{|w-2w| \left(\prod_{j=1}^{d_2} |w - cj^2 w| \right) \left(\prod_{j=1}^{d_1} \left( w - \left\lfloor  \frac{w}{cj^2}\right\rfloor\right)\right)}{ \left(w- \frac{w}{c \cdot i^2}\right) \cdot \left(2w -  \frac{w}{c \cdot i^2} \right) \left(\prod_{j=1}^{d_2} \left(w \cdot c \cdot j^2 - \frac{w}{c \cdot i^2}\right)\right) \prod_{j=1, j\neq i}^{d_1} \left|\left\lfloor \frac{w}{c \cdot i^2}\right\rfloor - \left\lfloor \frac{w}{c \cdot j^2}\right\rfloor\right|} \notag \\
& \leq  \frac{|w-2w| \left(\prod_{j=1}^{d_2} |w - cj^2 w| \right) \left(\prod_{j=1}^{d_1} \left( w -  \frac{w}{cj^2}\right)\right) \cdot \left(1+o(1)\right)}{ \left(w- \frac{w}{c \cdot i^2}\right) \cdot \left(2w -  \frac{w}{c \cdot i^2} \right) \left(\prod_{j=1}^{d_2} \left(w \cdot c \cdot j^2 - \frac{w}{c \cdot i^2}\right)\right) \prod_{j=1, j\neq i}^{d_1} \left|\left\lfloor \frac{w}{c \cdot i^2}\right\rfloor - \left\lfloor \frac{w}{c \cdot j^2}\right\rfloor\right|} \label{anotherlongequation}
\end{flalign}
Here, the final inequality used \Cref{alongequation}. 
Let us consider the expression 
$\prod_{j=1,j\neq i}^{d_{1}}\left\vert \left\lfloor \frac{%
w}{c\cdot i^{2}}\right\rfloor - \left\lfloor \frac{w}{c\cdot j^{2}}\right\rfloor \right\vert $%
. \ 
This quantity is \emph{at least}
\begin{align}
\prod_{j=1, j\neq i}^{d_1}\left( \left|\frac{w}{c \cdot i^2} - \frac{w}{c
\cdot j^2} \right| -1\right) = w^{d_1-1} \cdot \prod_{j=1, j \neq i}^{d_1}
\frac{\left|j^2 - i^2\right| - \frac{c i^2 j^2}{w}}{ci^2 j^2}  \notag \\
= w^{d_1-1} \cdot \prod_{j=1, j \neq i}^{d_1} \frac{\left|j-i| \cdot
|j+i\right| - \frac{c i^2 j^2}{w}}{ci^2 j^2}  \notag \\
= \left(\frac{w}{ci^2}\right)^{d_1-1} \cdot \prod_{j=1, j \neq i}^{d_1}
\frac{\left|j-i| \cdot |j+i\right| - \frac{c i^2 j^2}{w}}{j^2}
\label{aboveexp}
\end{align}

We claim that Expression \eqref{aboveexp} is at least
\begin{equation}
\left( \frac{w}{ci^{2}}\right) ^{d_{1}-1}\cdot \frac{1}{2}.
\label{provedintheappendix}
\end{equation}%
In the case that $c=2$ and $d_{1}$ is (at most) $w^{1/3}$, this is precisely
\cite[Claim 4]{zhandry}. \ We will ultimately take $c$ to be a constant
strictly greater than 2 and hence $d_{1}=\lfloor \left( w/c\right)
^{1/3}\rfloor $ is a constant factor smaller than $w^{1/3}$. \ The proof of
\cite[Claim 4]{zhandry} works with cosmetic changes in this case. \ For
completeness, we present a derivation of the claim in \Cref{app}.

\Cref{provedintheappendix} implies that Expression \eqref{anotherlongequation} is
at most:

\begin{align}
\frac{|w-2w| \left(\prod_{j=1}^{d_2} |w - cj^2 w| \right)
\left(\prod_{j=1}^{d_1} \left( w - \frac{w}{cj^2}\right)\right) \cdot
\left(1+o(1)\right)}{ \left(w- \frac{w}{c \cdot i^2}\right) \cdot \left(2w -
\frac{w}{c \cdot i^2} \right) \left(\prod_{j=1}^{d_2} \left(w \cdot c \cdot
j^2 - \frac{w}{c \cdot i^2}\right)\right) \left(\frac{w}{ci^2}%
\right)^{d_1-1} \cdot \frac{1}{2}}  \notag \\
= \frac{2 \left(\prod_{j=1}^{d_2} |1 - cj^2| \right) \left(\prod_{j=1}^{d_1}
\left(1 - \frac{1}{cj^2}\right)\right) \cdot \left(1+o(1)\right)}{ \left(1-
\frac{1}{c \cdot i^2}\right) \cdot \left(2 - \frac{1}{c \cdot i^2} \right)
\left(\prod_{j=1}^{d_2} \left(c \cdot j^2 - \frac{1}{c \cdot i^2}%
\right)\right) \left(\frac{1}{ci^2}\right)^{d_1-1}}  \notag \\
= \frac{2 \left(\prod_{j=1}^{d_2} (j^2-1/c) \right) \left(\prod_{j=1}^{d_1}
\left(1 - \frac{1}{cj^2}\right)\right) \cdot \left(1+o(1)\right)}{ \left(1-
\frac{1}{c \cdot i^2}\right) \cdot \left(2 - \frac{1}{c \cdot i^2} \right)
\left(\prod_{j=1}^{d_2} \left(j^2 - \frac{1}{c^2 \cdot i^2}\right)\right)
\left(\frac{1}{ci^2}\right)^{d_1-1}}  \notag \\
\leq \frac{2 \left(1+o(1)\right)}{ \left(1- \frac{1}{c \cdot i^2}\right)
\cdot \left(2 - \frac{1}{c \cdot i^2} \right) \left(\frac{1}{ci^2}%
\right)^{d_1-1}} \leq 4 \cdot \left(ci^2\right)^{d_1-1}.  \label{implication}
\end{align}

Summarizing Equations \eqref{verylongequation} and \eqref{implication}, we have shown
that: for any quantity $c \cdot i^2 \cdot w \in T_2$,
\begin{equation}  \label{keyeq1}
\left|\Phi(c \cdot w \cdot i^2)\right|/\left|\Phi(w)\right| \leq \frac{2
\left(1+o(1)\right)}{\left(1-\frac{1}{c \cdot i^2}\right) \cdot (c \cdot i^2
- 2) \cdot (c\cdot i^2)^{d_1}}
\end{equation}
and for any quantity $\left\lfloor \frac{w}{c \cdot i^2} \right\rfloor \in T_1$,
\begin{equation}  \label{keyeq2}
\left|\Phi\left( \left\lfloor \frac{w}{c \cdot i^2}\right\rfloor\right)\right|/\left|\Phi(w)%
\right| \leq 4\cdot \left(ci^2\right)^{d_1-1}.
\end{equation}

Let $\phi =\Phi /C$, where $C$ is as in \Cref{Cdef}. \ Let $%
D_{1}=d_{1}$ and $D_{2}=d_{2}$. \ \Cref{phdlem} implies that $\phi $ is a
feasible solution for the dual linear program of \Cref{duallpsec}. \ We now
show that, for any constant $\delta >0$, by choosing $c$ to be a
sufficiently large constant (that depends on $\delta $), we can ensure
that $\phi $ achieves objective value $1-2\delta $.

Let
\begin{eqnarray*}
A &=&|\Phi (w)|\cdot w^{D_{1}}, \\
B &=&|\Phi (2w)|\cdot (2w)^{D_{1}},
\end{eqnarray*}%
and
\begin{equation*}
E=\sum_{i=1}^{d_{1}}|\Phi (\lfloor w/ci^{2}\rfloor )|\cdot \left( \lfloor
w/ci^{2}\rfloor \right) ^{D_{1}}+\sum_{i=1}^{d_{2}}|\Phi (\lfloor w\cdot
ci^{2}\rfloor )|\cdot \left( w\cdot c\cdot i^{2}\right) ^{D_{1}}.
\end{equation*}%
By \Cref{zeros}, $C=A+B+E$.

Moreover, observe that $\text{sgn}(\Phi(w)) = - \text{sgn}(\Phi(2w))$, so
without loss of generality we may assume $\Phi(w) \geq 0$ and $\Phi(2w) \leq
0$ (if not, then replace $\Phi$ with $-\Phi$ throughout).

We now claim that by choosing $c$ to be a sufficiently large constant, we
can ensure that $E \leq \delta \cdot A$. To see this, observe that
Equations \eqref{keyeq1} and \eqref{keyeq2}, along with the fact that $%
D_1=d_1$ and $D_2=d_2$ implies that
\begin{align*}
E/A \leq \frac{1}{w^{D_1}} \left[ \left(\sum_{i=1}^{d_1} \left(\lfloor
w/ci^2\rfloor\right)^{D_1}\cdot 4 \cdot \left(ci^2\right)^{d_1-1} \right) +
\left(\sum_{i=1}^{d_2} \left(w \cdot c \cdot i^2\right)^{D_1} \frac{2
\left(1-\frac{1}{c \cdot i^2}\right)\left(1+o(1)\right)}{(c \cdot i^2 - 2)
\cdot (c\cdot i^2)^{d_1}}\right)\right] \\
\leq \frac{1}{w^{D_1}} \left[ \left(\sum_{i=1}^{d_1}
\left(w/ci^2\right)^{D_1} \cdot 4 \cdot \left(ci^2\right)^{d_1-1} \right) +
\left(\sum_{i=1}^{d_2} \left(w \cdot c \cdot i^2\right)^{D_1} \frac{2
\left(1-\frac{1}{c \cdot i^2}\right)\left(1+o(1)\right)}{(c \cdot i^2 - 2)
\cdot (c\cdot i^2)^{d_1}}\right) \right] \\
\leq 4 \left(\sum_{i=1}^{d_1} \frac{1}{c\cdot i^2}\right) +
\left(\sum_{i=1}^{d_2} \frac{2 \left(1+o(1)\right)}{\left(1-\frac{1}{c \cdot
i^2}\right) \left(c \cdot i^2 - 2\right)}\right) \\
\end{align*}

Since $\sum_{i=1}^{\infty}1/(ci^2) \leq \frac{\pi^2}{6c}$, we see that
choosing $c$ to be a sufficiently large constant depending on $\delta$
ensures that $E/A \leq \delta$ as desired.

Hence, $\phi$ achieves objective value at least
\begin{align*}
\phi(w) \cdot w^{D_1} - \phi(2w) \cdot (2w)^{D_1} - \sum_{\ell \in \{1,
\dots, N\}, \ell \not\in \{w, 2w\}} |\phi(\ell)| \cdot \ell^{D_1} \\
\geq \frac{A + B - E}{A+B+E} \geq \frac{(1-\delta) A + B}{%
(1+\delta)A + B} \geq 1-2\delta.
\end{align*}

\subsection{Approximate counting with classical samples}
\label{sec:classical_samples_queries}

For completeness, in this section, we sketch classical counterparts of \Cref{thm:main} and \Cref{thm:alg}. That is, we show tight bounds on classical randomized algorithms for $\ApxCount$ that make membership queries and have access to uniform random samples from the set being counted.

\begin{proposition}
There is a classical randomized algorithm that solves $\ApxCount$ with high probability using either $O(N/w)$ queries to the membership oracle for $S$, or else using $O(\sqrt{w})$ uniform samples from $S$.
\end{proposition}

\begin{proof}[Proof sketch]
By reducing approximate counting to the problem of estimating the mean of a biased coin, $O(N/w)$ queries are sufficient.

Alternatively, if we take $R$ samples, then the expected number of birthday collisions is $\binom{R}{2} \cdot \frac{1}{|S|}$ and the variance is $\binom{R}{2} \cdot \frac{1}{|S|}\left(1-\frac{1}{|S|}\right)$. So, taking $O(\sqrt{w})$ samples and computing the number of birthday collisions is sufficient to distinguish $|S| \le w$ from $|S| \ge 2w$ with $\frac{2}{3}$ success probability.
\end{proof}

\begin{proposition}
Let $M$ be a classical randomized algorithm that makes $T$ queries to the membership oracle for $S$,
and takes a total of $R$ uniform samples from $S$.
If $M$ decides whether $\left\vert S\right\vert =w$\ or $\left\vert S\right\vert =2w$
with high probability, promised that one of those is the case, then either
$T = \Omega(N/w)$ or $R = \Omega(\sqrt{w})$.
\end{proposition}

\begin{proof}[Proof sketch]
Note that without loss of generality, we may
assume that the algorithm first takes all of the samples it needs, and then queries random elements of $[N]$ that did not appear in the samples. Suppose the algorithm takes $R = o(\sqrt{w})$ samples and then makes $T = o(N/w)$ queries. Consider what happens when the algorithm tries to distinguish a random subset of size $w$ from a random subset of size $2w$ of $[N]$. By a union bound, the probability that the algorithm sees any collisions in the samples is $o(1)$, and the probability that the algorithm finds any additional elements of $S$ via queries is also $o(1)$. So, if the set has size either $w$ or $2w$, with $1 - o(1)$ probability, the algorithm's view of the samples is just a random subset of size $R$ of $[N]$ drawn without replacement, and the algorithm's view of the queries is just $T$ ``no'' answers to membership queries. Hence, the algorithm fails to distinguish random sets of size $w$ and size $2w$ with any constant probability of success.
\end{proof}

\subsection{Extending the lower bound to QSampling unitarily}
\label{sec:unitary}
So far in this section we have proved upper and lower bounds on the power of quantum algorithms for approximate counting that have 
access to two resources (in addition to membership queries): 
copies of $|S\>$, and the unitary transformation that reflects about $|S\>$. 
The assumption of access to the reflection unitary is justified by the argument that, if we had access to a unitary that prepared $|S\>$, then it could be used to reflect about $|S\>$ as well. 

Giving the algorithm access 
to just the two resources above is an appealing model to use for upper bounds, since it does not assume anything about the method by which copies of $|S\>$ are prepared. This means algorithms
derived in this model work in many different settings.
For example, the algorithm may be able to QSample
because someone else simply handed the algorithm copies of $|S\>$, or perhaps several copies of $|S\>$ just happen to be stored in the algorithm's quantum memory as a side effect of the execution of some earlier quantum algorithm. The upper bound given
in \Cref{thm:alg} applies in any of these settings.

On the other hand, since only
permitting access to QSamples
and reflections about $|S\>$ ties the algorithm's hands,
lower bounds for this model
(e.g., \Cref{thm:main}) could be viewed
as weaker than is desirable. 
In particular, our original 
justification
for allowing access to reflections about $|S\>$
was that access to a unitary that prepared 
the state $|S\>$ would in particular allow 
such reflections to be done. 
Given this justification, it is very natural to wonder whether 
our lower bounds extend beyond just 
QSamples and reflections, to
algorithms that are given access
to \emph{some} unitary process that 
permits both QSampling and reflections about $|S\>$. 

Note that
an algorithm with access to such a unitary could potentially exploit the unitary in ways other
than QSamples and reflections to learn information about $|S\>$. 
For example, the algorithm could choose to run the unitary on inputs that do not produce $|S\>$.
More generally, given a quantum circuit that implements a unitary, it is possible to construct, in a completely black-box manner, the inverse of this unitary, and also a controlled version of the unitary.
The algorithm may choose to run the inverse on a state other than $|S\>$ to learn some additional information that is not captured by access
to QSamples and reflections alone.

In summary, in this section
we ask whether we can we extend the lower bound of \Cref{thm:main} to work in a model where the algorithm is given access to some unitary operator that conveys the power to both QSample and reflect about $|S\>$.\footnote{We thank Alexander Belov (personal communication) for raising this question.} 
Via \cref{thm:finalthm} below, we explain that the answer is yes.

\medskip

It may seem convenient to assume that the unitary
transformation preparing $|S\>$ 
maps the all-zeros state to $|S\>$.
But this is not the most general method of preparing $|S\>$ by a unitary.
A unitary $U$ that maps the all-zeros state to $|S\>|\psi\>$ would also suffice to create copies of $|S\>$, since the register containing $|\psi\>$ can simply be ignored for the remainder of the computation.
More formally, assume $U$ behaves as 
\begin{equation}
    U|0^m\> = |S\>|\psi\>,
\end{equation}
where $|S\>|\psi\>$ is some $m$-qubit state.
Clearly we can use $U$ to create as many copies of $|S\>$ as we like, 
which as a by-product also creates copies of $|\psi\>$.
This unitary also lets us reflect about $|S\>$. 
To see how, first use this unitary to create a copy of $|\psi\>$, and then consider the action of the unitary $U(\id - 2|0^m\>\<0^m|)U^\dagger$ on the state $|\phi\>|\psi\>$ for any state $|\phi\>$. We claim that this unitary acts as a reflection about $|S\>$ when restricted to the first register. 
This establishes that any $U$ of this form subsumes the power of both QSamples and reflections about $|S\>$.

Let us also assume without loss of generality that $|S\>|\psi\>$ is orthogonal to $|0^m\>$ from now on. This can be achieved by adding an additional qubit to the input that is always negated by the unitary. That is, we could instead consider the map $(U \otimes X)|0^m\>|0\> = |S\>|\psi\>|1\>$, which is orthogonal to the starting state by construction, and only increases the value of $m$ by $1$.

Of course, the requirement that $U|0^m\> = |S\>|\psi\>$ does not fully specify $U$, as
it does not prescribe how $U$ behaves on other input states.
A reasonable prescription is that $U$ should behave ``trivially'' on other input states, so that it does not leak information about $S$ by its behavior on other states. 
In tension with this prescription is the fact the rest of the unitary must depend on $S$, since the first column of the unitary contains $|S\>$, and the rest of the columns have to be orthogonal to this.

Alexander Belov (personal communication) brought to our attention a 
very simple construction of such a unitary that leaks minimal additional information about $S$. 
Consider the unitary $U$ that satisfies $U|0^m\> = |S\>|\psi\>$ and $U|S\>|\psi\> = |0^m\>$, with $U$ acting as identity outside $\mathrm{span}\{|0^m\>,|S\>|\psi\>\}$. $U$ is simply a reflection about the state $\frac{1}{\sqrt{2}}\bigl(|0^m\>-|S\>|\psi\>\bigr)$. This state is correctly normalized because we assumed that $|S\>|\psi\>$ is orthogonal to $|0^m\>$. 
Clearly $U$ is now fully specified on the entire domain (once we have fixed $|\psi\>$) and it does not seem to leak any additional information about $S$.

In order to prove concrete lower bounds on the cost
of algorithms for approximate counting given access to $U$, we need to fix $|\psi\>$. 
To answer the question posed in this section, we only need to establish that there exists \emph{some} choice of $|\psi\>$ for which our algorithms cannot be improved. (Note that we cannot hope to establish lower bounds for arbitrary $|\psi\>$, since $|\psi\>$ could just contain the answer to the problem we are solving.)

To this end we make the specific choice of $|\psi\> = |S\>$ and consider the unitary $V$ that acts as the unitary $U$ above with $|\psi\> = |S\>$. 
In other words, $V$ maps $|0^m\>$ to $|S\>|S\>$, $|S\>|S\>$ to $|0^m\>$, and acts as identity on the rest of the space. 
We also assume that $|0^m\>$ is orthogonal to $|S\>|S\>$. 
In other words, $V$ simply reflects about the state $\frac{1}{\sqrt{2}}\bigl(|0^m\>-|S\>|S\>\bigr)$.

As previously discussed, granting an algorithm access to this unitary $V$ lends the algorithm at least as much power the ability to QSample and perform reflections about $|S\>$. How efficiently can we solve approximate counting with membership queries and uses of the unitary $V$? 

We can use our Laurent polynomial method to establish optimal lower bounds in this model as well and we obtain lower bounds identical to \Cref{thm:main}. 

\begin{theorem} \label{thm:finalthm} \label{thm:unitary}
Let $Q$ be a quantum algorithm that makes $T$ queries to the membership oracle for $S$,
and makes $R$ uses of the unitary $V$ defined above (and its inverse and controlled-$V$).
If $Q$ decides whether $\left\vert S\right\vert =w$\ or $\left\vert S\right\vert =2w$
with high probability, promised that one of those is the case, then either
\begin{equation}
T=\Omega\left(  \sqrt{\frac{N}{w}}\right)  \qquad \textrm{or} \qquad
R=\Omega\left(  \min\left\{  w^{1/3},\sqrt{\frac{N}{w}}\right\}  \right).
\end{equation}
\end{theorem}
\begin{proof}
We follow the same strategy as in
the proof of \Cref{thm:main}. Recall that $x \in \{0, 1\}^N$ denotes the indicator vector of the set $S$. We only need to show that such a quantum algorithm gives rise to a Laurent polynomial in $|S|:=\sum_{i=1}^n x_i$, with maximum exponent $O(T+R)$ and minimum exponent at least $-O(R)$ (as shown in \Cref{laurentlem} for the QSamples and reflections model).

We can prove this exactly the same way as \Cref{laurentlem} is established. Our quantum algorithm starts out from a canonical starting state that does not depend on the input and hence each entry of the starting state is a degree-$0$ polynomial.
Membership queries involve multiplication with an oracle whose entries are ordinary polynomials of degree at most $1$. The only thing that remains is understanding what the entries of the unitary $V$ look like. We claim that the entries of $V$ are given
by a polynomial of degree at most 2 in the entries of the input $x$,
with all coefficients of this degree-2 polynomial equal
to either a constant, or a constant multiple of $|S|^{-1}$.

To see this, note that $V$ is simply a reflection about the state
\begin{equation}
    \frac{1}{\sqrt{2}}\bigl(|0^m\>-|S\>|S\>\bigr)
    = \frac{1}{\sqrt{2}}\left(|0^m\>-\frac{1}{|S|}\Bigl(\sum_{i}x_i|i\>\Bigr)\Bigl(\sum_{j}x_j|j\>\Bigr)
    \right).
\end{equation} 
The coefficient in front of $|0^m\>$ is a degree-$0$ polynomial and the other nonzero coefficients are
a polynomial of degree at most 2 in the entries of the input $x$, with each coefficient of this polynomial equal to a constant multiple of $|S|^{-1}$.

Hence, each entry of the unitary $V$ is also a polynomial of degree at most 2 in the entries of the input $x$,
with each coefficient of this degree-2 polynomial equal
to either a constant, or a constant multiple of $|S|^{-1}$. The same also holds for controlled-$V$, since that unitary is just the direct sum of identity with $V$. $V$ is also self-inverse, so we do not need to account for that separately.

After the algorithm has made all the membership queries and uses of $V$, each amplitude of the final quantum state can be expressed as a polynomial of degree $O(T + R)$ in the input $x$, in which all coefficients are constant multiples of $|S|^{-R}$. The acceptance probability $p(x)$ of this algorithm will be a sum of squares of such polynomials. 
Exactly
as in the proof of \Cref{thm:main},
\Cref{symlem} implies that there is a univariate polynomial
$q$ of degree at most $O(T+R)$, with coefficients that are multiples
of the coefficients of $p$, such that for all integers $k \in \{0, \dots, N\}$, 
\begin{equation}
q\left(  k\right)  :=\E_{\left\vert X\right\vert =k}\left[
p\left(  X\right)  \right]  .
\end{equation}
Since the coefficients of $p(X)$ are constant multiples of $|X|^{-2R}$, 
$q$ is in fact a real Laurent polynomial in $k$, with maximum exponent at most $O(R+T)$\ and minimum exponent at
least $-2R$.
The theorem follows by a direct 
application \Cref{mainthm} to $q$.
\end{proof}
\section{Discussion and open problems\label{OPEN}}

\label{sec:open}

\subsection{Approximate counting with QSamples and queries only}

\label{IMPROVE}If we consider the model where we only have membership
queries and samples (but no reflections), then the best upper bound we can
show is $O\left( \min \left\{ \sqrt{w},{\sqrt{{N}/{w}}}\right\} \right) $,
using the sampling algorithm that looks for birthday collisions, and the 
quantum counting algorithm. It
would be interesting to improve the lower bound further in this case, but it
is clear that the Laurent polynomial approach cannot do so, since it hits a
limit at $w^{1/3}$. \ Hence a new approach is needed to tackle the model
without reflections.

We now give what we think is a viable path to solve this problem. \
Specifically, we observe that our problem---of lower-bounding the number of
copies of $\left\vert S\right\rangle $\ \textit{and} the number of queries
to $\mathcal{O}_{S}$\ needed for approximate counting of $S$---can be
reduced to a pure problem of lower-bounding the number of copies of $%
\left\vert S\right\rangle $. \ To do so, we use a hybrid argument, closely
analogous to an argument recently given by Zhandry \cite{zhandry:lightning}
in the context of quantum money.

Given a subset $S\subseteq\left[ L\right] $, let $\left\vert S\right\rangle $%
\ be a uniform superposition over $S$ elements. \ Then let%
\begin{equation}
\rho_{L,w,k}:=\E_{S\subseteq\left[ L\right] ~:~\left\vert S\right\vert =w}%
\left[ \left( \left\vert S\right\rangle \left\langle S\right\vert \right)
^{\otimes k}\right]
\end{equation}
be the mixed state obtained by first choosing $S$\ uniformly at random
subject to $\left\vert S\right\vert =w$, then taking $k$ copies of $%
\left\vert S\right\rangle $. \ Given two mixed states $\rho$\ and $\sigma$,
recall also that the \textit{trace distance}, $\left\Vert
\rho-\sigma\right\Vert _{\mathrm{tr}}$, is the maximum bias with which $%
\rho $\ can be distinguished from $\sigma$\ by a single-shot measurement.

\begin{theorem}
\label{hybridthm}Let $2w\leq L\leq N$. \ Suppose $\left\Vert
\rho_{L,w,k}-\rho_{L,2w,k}\right\Vert _{\mathrm{tr}}\leq\frac{1}{10}$. \
Then any quantum algorithm $Q$\ requires either $\Omega\left( \sqrt{\frac{N}{%
L}}\right) $\ queries to $\mathcal{O}_{S}$\ or else $\Omega\left( k\right) $%
\ copies of $\left\vert S\right\rangle $ to decide whether $\left\vert
S\right\vert =w$\ or $\left\vert S\right\vert =2w$ with success probability
at least $2/3$, promised that one of those is the case.
\end{theorem}

\begin{proof}
Choose a subset $S\subseteq\left[  N\right]  $ uniformly at random, subject to
$\left\vert S\right\vert =w$\ or $\left\vert S\right\vert =2w$, and consider
$S$ to be fixed. \ Then suppose we choose $U\subseteq\left[  N\right]
$\ uniformly at random, subject to both $\left\vert U\right\vert =L$\ and
$S\subseteq U$. \ Consider the hybrid in which $Q$ is still given $R$\ copies
of the state $\left\vert S\right\rangle $, but now gets oracle access to
$\mathcal{O}_{U}$\ rather than $\mathcal{O}_{S}$. \ Then so long as $Q$ makes
$o\left(  \sqrt{\frac{N}{L}}\right)  $\ queries to its oracle, we claim that
$Q$ cannot distinguish this hybrid from the \textquotedblleft
true\textquotedblright\ situation (i.e., the one where $Q$\ queries
$\mathcal{O}_{S}$) with $\Omega\left(  1\right)  $\ bias. \ This claim follows
almost immediately from the BBBV Theorem \cite{bbbv}. \ In effect, $Q$ is
searching the set $\left[  N\right]  \setminus S$ for any elements of
$U\setminus S$\ (the \textquotedblleft marked items,\textquotedblright\ in
this context), of which there are $L-\left\vert S\right\vert $ scattered
uniformly at random. \ In such a case, we know that $\Omega\left(  \sqrt
{\frac{N-\left\vert S\right\vert }{L-\left\vert S\right\vert }}\right)
=\Omega\left(  \sqrt{\frac{N}{L}}\right)  $\ quantum queries are needed to
detect the marked items with constant bias.

Next suppose we first choose $U\subseteq\left[  N\right]  $ uniformly at
random, subject to $\left\vert U\right\vert =L$, and consider $U$ to be fixed.
\ We then choose $S\subseteq U$\ uniformly at random, subject to $\left\vert
S\right\vert =w$\ or $\left\vert S\right\vert =2w$. \ Note that this produces
a distribution over $\left(  S,U\right)  $\ pairs identical to the
distribution that we had above. \ In this case, however, since $U$ is fixed,
queries to $\mathcal{O}_{U}$\ are no longer relevant. \ The only way to decide
whether\ $\left\vert S\right\vert =w$\ or $\left\vert S\right\vert =2w$\ is by
using our copies of $\left\vert S\right\rangle $---of which, by assumption, we
need $\Omega\left(  k\right)  $\ to succeed with constant bias, even after
having fixed $U$.
\end{proof}

One might think that \Cref{hybridthm} would lead to immediate improvements
to our lower bound for the queries and samples model. \ In practice,
however, the best lower bounds that we currently have, even purely on the
number of copies of $\left\vert S\right\rangle $, come from the Laurent
polynomial method (\Cref{thm:main})! \ Having said that, we are optimistic
that one could obtain a lower bound that beats \Cref{thm:main}\ at least
when $w$\ is small, by combining \Cref{hybridthm} with a brute-force
computation of trace distance.

\subsection{Approximate counting to multiplicative factor
\texorpdfstring{$1+\eps$}{1+eps}}

Throughout, we considered the task of approximating $\left\vert S\right\vert
$ to within a multiplicative factor of $2$. \ But suppose our task was to
distinguish the case $\left\vert S\right\vert \leq w$\ from the case $%
\left\vert S\right\vert \geq \left( 1+\varepsilon \right) w$; then what is
the optimal dependence on $\varepsilon $?

In the model with quantum membership queries only, the algorithm 
of Brassard et al.~\cite[Theorem 15]{BHMT02}
makes $O\Bigl(\frac{1}{\varepsilon }\sqrt{\frac{N}{w}}\Bigr)$\ queries, which is optimal~\cite{nayak-wu}. 
\ The algorithm uses amplitude amplification, the basic primitive of Grover's search algorithm
\cite{grover}. \ The original algorithm of Brassard et al.\ also used
quantum phase estimation, in effect \textit{combining} Grover's algorithm
with Shor's period-finding algorithm. \ However, one can remove the phase estimation, and adapt Grover search with an unknown  number of marked items to get an approximate count of the number of marked  items~\cite{aaronson-rall}.

One can also show without too much difficulty that in the 
queries+QSamples model, the problem can be solved with 
\begin{equation}
O\left( \min \left\{ \frac{\sqrt{w}}{\varepsilon^2 },\frac{1}{\varepsilon}\sqrt{\frac{N}{w}}\right\} \right)
\end{equation}
queries and copies of $\left\vert S\right\rangle $. \ 
As observed after \Cref{thm:alg}, the problem can also be solved with
\begin{equation}
O\left( \min \left\{ \frac{w^{1/3}}{\varepsilon ^{2/3}},\frac{1}{\varepsilon
}\sqrt{\frac{N}{w}}\right\} \right)
\end{equation}%
samples and reflections. \ On the lower bound side, what generalizations of %
\Cref{thm:main}\ can we prove that incorporate $\varepsilon $? \ We note
that the explosion argument doesn't automatically\ generalize; one would
need to modify something to continue getting growth in the polynomials $u$\
and $v$\ after the first iteration. \ The lower bound using dual polynomials
should generalize, but back-of-the-envelope calculations show that the lower
bound does not match the upper bound.

\subsection{Other questions}

\para{Non-oracular example of our result.} \
Is there any interesting real-world example of a class of sets for which
QSampling and membership testing are both efficient, but approximate
counting is not? \ (I.e., is there an interesting non-black-box setting that
appears to exhibit the behavior that this paper showed can occur in the
black-box setting?)

\para{The Laurent polynomial connection.} \ At
a deeper level, is there is any meaningful connection between our two uses
of Laurent polynomials? {\ And what other applications can be found for the
Laurent polynomial method?}

\section{Followup work}

Since this work was completed, Belovs and Rosmanis \cite{belovs} obtained essentially tight lower bounds on the complexity of approximate counting with access to membership queries, QSamples, reflections, and a unitary transformation that prepares the QSampling state, for all possible tradeoffs between these different resources. Additionally, they resolve the $\eps$-dependence of approximate counting to multiplicative factor $1 + \eps$. The techniques involved are quite different from ours: Belovs and Rosmanis use a generalized version of the quantum adversary bound that allows for multiple oracles, combined with tools from the representation theory of the symmetric group.

\section*{Acknowledgments}

We are grateful to many people: Paul Burchard, for suggesting the problem of
approximate counting with queries and QSamples;\ MathOverflow user
\textquotedblleft fedja\textquotedblright\ for letting us include 
\Cref{fedjatight} and \Cref{fedjalem};\ Ashwin Nayak, for extremely helpful
discussions, and for suggesting the transformation of linear programs used
in our extension of the method of dual polynomials to the Laurent polynomial
setting; Thomas Watson, for suggesting the intersection approach to proving
an $\mathsf{SBP}$ vs. $\mathsf{QMA}$ oracle separation; and Patrick Rall,
for helpful feedback on writing. JT would particularly like to thank Ashwin 
Nayak for his warm hospitality and deeply informative discussions during a visit to Waterloo.

\appendix

\section{Establishing Equation \protect\ref*{provedintheappendix}}

\label{app}

\subsection{A clean calculation establishing a loose version of equation
\protect\ref*{provedintheappendix}}

For clarity of exposition, we begin by presenting a relatively clean
calculation that establishes a slightly loose version of 
\Cref{provedintheappendix}. Using just this looser bound, we would be able to
establish that \Cref{provedintheappendix} holds (with the constant $1/2$
replaced by a slightly smaller constant) so long as we set $d_1$ to be $%
\Theta\left(w^{1/3}/\log w \right)$. A slightly more involved calculation
(cf. \Cref{tightsec}) is required to establish 
\Cref{provedintheappendix} for our desired value of $d_1= \lfloor (w/c)^{1/3}
\rfloor$.

Expression \eqref{aboveexp} equals
\begin{align}
\left(\frac{w}{ci^2}\right)^{d_1-1} \cdot \frac{i^2}{\left((d_1)!\right)^2}
\cdot \prod_{j=1, j \neq i}^{d_1}\left( \left|j-i| \cdot |j+i\right| - \frac{%
c i^2 j^2}{w}\right)  \notag \\
= \left(\frac{w}{ci^2}\right)^{d_1-1} \cdot \frac{i^2}{\left((d_1)!\right)^2}
\cdot \prod_{j=1, j \neq i}^{d_1}\left( \left|j-i| \cdot |j+i\right| \right)
\cdot \left(1-\frac{c i^2 j^2}{w \cdot |j-i| |j+i|}\right)  \notag \\
= \left(\frac{w}{ci^2}\right)^{d_1-1} \cdot \frac{(d_1+i)! (d_1-i)!}{%
2\left((d_1)!\right)^2} \cdot \prod_{j=1, j \neq i}^{d_1} \left(1-\frac{c
i^2 j^2}{w \cdot |j-i| |j+i|}\right)  \label{looseversion} \\
\geq \left(\frac{w}{ci^2}\right)^{d_1-1} \cdot \frac{1}{2} \cdot \prod_{j=1,
j \neq i}^{d_1} \left(1-\frac{c i^2 j^2}{w \cdot |j-i| |j+i|}\right)  \notag
\\
\geq \left(\frac{w}{ci^2}\right)^{d_1-1} \cdot \frac{1}{2} \cdot
\left(1-\sum_{j=1, j \neq i }^{d_1}\frac{c i^2 j^2}{w \cdot |j-i| |j+i|}%
\right)  \notag \\
\geq \left(\frac{w}{ci^2}\right)^{d_1-1} \cdot \frac{1}{2} \cdot \left(1-
\frac{c i^2}{w}\sum_{j=1, j \neq i }^{d_1}\frac{j^2}{ |j-i| |j+i|}\right).
\label{loosemiddle}
\end{align}

Let us consider the expression $\sum_{j=1, j \neq i }^{d_1}\frac{j^2}{ |j-i|
|j+i|}$. If $i^2 \not \in [j^2/2, 3j^2/2]$, then the $j$'th term in this sum
is at most $2$. Hence, letting $H_i$ denote the $i$th Harmonic number and
using the fact that $H_i \leq \ln(i+1)$,
\begin{align}
& \sum_{j=1, j \neq i }^{d_1}\frac{j^2}{ |j-i| |j+i|}  \notag \\
& \leq 2 \cdot d_1 + \sum_{j = \lfloor \sqrt{2/3} i \rfloor}^{\lfloor \sqrt{2%
} i \rfloor} \frac{j^2}{|j-i| |j+i|}  \notag \\
& \leq 2 \cdot d_1 + \sum_{j = \lfloor \sqrt{2/3} \cdot i \rfloor}^{\lceil
\sqrt{2} \cdot i \rceil} \frac{j}{|j-i|}  \notag \\
& \leq 2 d_1 + \sqrt{2} \cdot i \cdot \sum_{j=1}^{(\sqrt{2}-1) \cdot i} 2/j
\notag \\
& \leq 2 d_1 + 2 \sqrt{2} \cdot i \cdot H_{i} \leq 2d_1 + 2 \sqrt{2} i
\ln(i+1).  \label{eq:ieq}
\end{align}

We conclude that if $d_1$ were set to a value less than $w^{1/3}/(100 \cdot
c^2 \cdot \ln(w))$ (rather than to $\lfloor \left(w/c\right)^{1/3}\rfloor$),
then Expression \eqref{loosemiddle} is at least
\begin{equation}
\left(\frac{w}{ci^2}\right)^{d_1-1} \cdot \frac{1-1/c}{2}.
\end{equation}

\subsection{The tight bound}

\label{tightsec} To obtain the tight bound, we need a tighter sequence of
inequalities following Expression \eqref{looseversion}. Specifically,
Expression \eqref{looseversion} is bounded below by:

\begin{align}
\geq \left(\frac{w}{ci^2}\right)^{d_1-1} \cdot \frac{1}{2} \left(1+ \frac{i}{%
2d_1}\right)^i \cdot \prod_{j=1, j \neq i}^{d_1} \left(1-\frac{c i^2 j^2}{w
\cdot |j-i| |j+i|}\right)  \notag \\
\geq \left(\frac{w}{ci^2}\right)^{d_1-1} \cdot \frac{1}{2} \cdot
e^{i^2/(2d_1)} \cdot \prod_{j=1, j \neq i}^{d_1} \left(1-\frac{c i^2 j^2}{w
\cdot |j-i| |j+i|}\right)  \notag \\
\geq \left(\frac{w}{ci^2}\right)^{d_1-1} \cdot \frac{1}{2} \cdot
e^{i^2/(2d_1)} \cdot \prod_{j=1, j \neq i}^{d_1} \left(1-\frac{c i^2 j^2}{w
\cdot |j-i| |j+i|}\right)  \label{tightbegin}
\end{align}

The rough idea of how to proceed is as follows. \Cref{eq:ieq}
implies that for $i \ll w^{1/3}/\ln w$, the factor
\begin{equation*}
F_1 := \prod_{j=1, j \neq i}^{d_1} \left(1-\frac{c i^2 j^2}{w \cdot |j-i|
|j+i|}\right)
\end{equation*}
is at some a positive constant, and hence Expression \eqref{tightbegin} is
bounded below by the desired quantity. If $i \gtrsim w^{1/3}/\ln w$, then
\Cref{eq:ieq} does not yield a good bound on this factor, leaving
open the possibility that this factor is subconstant. But in this case, the
factor $F_2:= e^{i^2/(2d_1)} \geq e^{\tilde{\Omega}(d_1)}$, and the
largeness of $F_2$ dominates the smallness of $F_1$.

In more detail, let $x_{i, j} = \frac{ci^2j^2}{w \cdot |i-j| j+i||}$. Then
for all $i \neq j$ such that $i, j \leq d_1$,
\begin{equation}  \label{xijbound}
x_{i,j} \leq \frac{c \cdot d_1^2 (d_1-1)^2}{(2d_1-1) \cdot w} \leq \frac{c
\cdot d_1^3}{2w} \leq 1/2,
\end{equation}
where in the final inequality we used the fact that $d_1 \leq (w/c)^{1/3}$.

Using the fact that $1-x \geq e^{-x-x^2}$ for all $x \in [0, 1/2]$, we can
write
\begin{equation*}
F_1 \geq \prod_{j=1, j\neq i}^{d_1} e^{-x_{i,j} - x_{i, j}^2}.
\end{equation*}
Hence,
\begin{align*}
F_1 \cdot F_2 \geq \exp\left(i^2/(2d_1) - \sum_{j=1, j\neq i}^{d_1} -x_{i,j}
- x_{i, j}^2\right).
\end{align*}

From Equations \eqref{eq:ieq} and \eqref{xijbound}, we know that
\begin{align*}
\sum_{j=1, j\neq i}^{d_1} x_{i,j} + x_{i, j}^2 \leq \frac{c i^2}{w} \cdot
\left(3d_1 + 3 \sqrt{2} i \ln(i+1)\right) \leq \frac{ci^2}{w} \cdot \left(4
d_1 \ln(d_1)\right).
\end{align*}
Hence,
\begin{align*}
F_1 \cdot F_2 & \geq \exp\left(i^2/(2d_1) - \frac{ci^2}{w} \cdot 4 c
\ln(d_1)\right) \\
&= \exp\left(i^2 \left(\frac{1}{2d_1} - \frac{4c^2 \ln(d_1)}{w}\right)\right)
\\
& \geq \exp\left(i^2 \cdot \frac{1}{2d_1} \cdot (1-o(1))\right) \\
& \geq 1.
\end{align*}

\Cref{provedintheappendix} follows.


\phantomsection\addcontentsline{toc}{section}{References}
\bibliographystyle{alphaurl}
\bibliography{apxcount}

\end{document}